\documentclass[10pt]{amsart}

\usepackage[ansinew]{inputenc}

\usepackage{textcomp}

\usepackage{cite}

\usepackage{amsthm}
\usepackage{amssymb}
\usepackage{amsfonts}

\newtheorem{thm}{Theorem}
\newtheorem{lem}[thm]{Lemma}
\newtheorem{cor}[thm]{Corollary}

\newtheorem{pro}[thm]{Proposition}

\theoremstyle{remark}
\newtheorem*{rem}{Remark}

\newcommand{\lk}{\left(}
\newcommand{\rk}{\right)}
\newcommand{\vo}{|\Omega|}
\newcommand{\s}{\sigma}
\newcommand{\g}{\gamma}
\newcommand{\R}{\mathbb{R}}
\newcommand{\La}{\Lambda}
\newcommand{\G}{\tilde{\Gamma}}
\newcommand{\B}{\tilde{B}}

\title{Universal Bounds for Traces of the Dirichlet Laplace Operator}
\author{Leander Geisinger \ \& \ Timo Weidl}

\begin{document}

\begin{abstract}
We derive upper bounds for the trace of the heat kernel $Z(t)$ of the Dirichlet Laplace operator in an open set $\Omega \subset \R^d$, $d \geq 2$. In domains of finite volume the result improves an inequality of Kac. Using the same methods we give bounds on $Z(t)$ in domains of infinite volume.

For domains of finite volume the bound on $Z(t)$ decays exponentially as $t$ tends to infinity and it contains the sharp first term and a correction term reflecting the properties of the short time asymptotics of $Z(t)$.
To prove the result we employ refined Berezin-Li-Yau inequalities for eigenvalue means.
\end{abstract}

\maketitle

\section{Introduction and main results} \label{intro}

Let $\Omega$ be an open subset of $\mathbb{R}^d$, $d \geq 2$.  Consider the Laplace operator $-\Delta_\Omega$ on $\Omega$ 
subject to Dirichlet boundary conditions defined in the form sense on the form domain $H^1_0(\Omega)$. If the embedding $H^1_0(\Omega) \hookrightarrow L^2(\Omega)$ is compact, e.g. if the volume of $\Omega$ is finite, 
the spectrum of $-\Delta_\Omega$ is discrete and consists of a monotone sequence of positive eigenvalues 
$0 < \lambda_1 \leq \lambda_2 \leq \lambda_3 \leq \dots$ 
accumulating at infinity. We count these eigenvalues according to their  multiplicity. 

The main goal of this paper is to derive some new universal upper bounds for the trace of the heat kernel
\begin{displaymath}
Z(t) \,= \,\textnormal{Tr} \lk e^{+\Delta_\Omega t} \rk \,= \,\sum_k e^{-\lambda_k t}
\end{displaymath}
which are valid for arbitrary open sets $\Omega \subset \mathbb{R}^d$ with finite volume $\vo$ and for all $t > 0$.
The first and most fundamental bound of this type is due to M. Kac, \cite{Kac02}.  He proved that for any open domain
$\Omega \subset \R^d$ and all $t > 0$ the estimate
\begin{equation} \label{basic}
 Z(t) \,\leq \,\frac{\vo}{\left( 4 \pi t \right)^{\frac{d}{2}}}
\end{equation}
holds true. This bound is sharp in the sense that it reflects the leading term of the short time asymptotics of the function $Z(t)$, 
see \cite{Min, Kac01}
\begin{equation}\label{ztasymp}
 Z(t) \,= \,\frac{\vo}{\left( 4 \pi t \right)^{\frac{d}{2}}}\quad \mbox{as} \quad t\to 0+.
\end{equation}
%
Several improvements of (\ref{basic}) are known, e.g.  see 
\cite{Ber01,FLV,Dav02,Dav01,Sim01,Ber03} and further references therein.  
For example, M. van den Berg proved in \cite{Ber02}, that if $\Omega$ is a connected 
region with a smooth boundary $\partial \Omega$ and a surface area $|\partial \Omega|$, then
\begin{equation*} 
\left| Z(t) - \frac{\vo}{(4 \pi t)^{\frac{d}{2}}} + \frac{|\partial \Omega|}{4 \,(4 \pi t)^{\frac{d-1}{2}}} \right| \,\leq \,\frac{d^4}{\pi^{\frac{d}{2}}} \frac{\vo}{t^{\frac d2 -1}R^2}, \quad t>0,
\end{equation*}
where the constant $R$ depends on properties of $\partial \Omega$. 
This estimate contains even the second term of the short time asymptotic expansion of $Z(t)$,
see \cite{McSi,Smi,BrCa} and \cite{Bro}. 
Most of these results are based on a probabilistic approach and implement local estimates for the heat kernel.
Therefore one has to impose appropriate conditions on $\Omega$ and on its boundary $\partial \Omega$.

We use a different approach based on some refined spectral estimates
for the Riesz means
\begin{displaymath}
R_\sigma(\La) \,= \,\textnormal{Tr} \lk -\Delta_\Omega - \La \rk_-^\sigma \,= \,\sum_{k} \lk \La - \lambda_k \rk_+^\sigma,
\quad\Lambda>0.
\end{displaymath}
For these objects the fundamental bounds are given by the Berezin-Li-Yau inequalities
\begin{equation} \label{beliyau}
R_\sigma(\La) \,\leq \,L^{cl}_{\sigma,d} \,\vo \,\Lambda^{\sigma+\frac{d}{2}},\quad \sigma\geq 1 \, , \ \Lambda>0,
\end{equation} 
where 
\begin{displaymath}
L^{cl}_{\sigma,d} \,= \,\frac{\Gamma(\sigma +1)}{(4\pi)^{\frac{d}{2}} \Gamma\lk\sigma+\frac{d}{2}+1\rk}.
\end{displaymath}
This result is sharp as well in the sense that the bound captures the first term of the high energy asymptotics
\begin{equation*}
R_\sigma(\La) \,= \,L^{cl}_{\sigma,d} \,\vo \,\La^{\sigma + \frac{d}{2}}  + o \lk \La^{\sigma+ \frac{d}{2}} \rk
\quad\mbox{as}\quad \Lambda\to +\infty.
\end{equation*}
Via Laplace transformation - and reversely via Tauberian theorems -
this asymptotic formula is closely connected with \eqref{ztasymp}.
On the level of uniform inequalities one can deduce Kac' inequality on $Z(t)$ from 
Berezin-Li-Yau bounds. Reversely, to recover sharp Berezin-Li-Yau bounds from Kac' inequality
one needs some additional information. For example, in \cite{HaHe01} Harrell and Hermi formally  deduced
Berezin-Li-Yau bounds for $\sigma\geq 2$ from Kac' inequality
based on a monotonicity result by Harrell and Stubbe.\footnote{One should mention, that in fact,
due to Weyl's asymptotic law, the monotonicity result implies sharp Berezin-Li-Yau bounds for $\sigma\geq 2$ on its own.}
Similar arguments fail for $\sigma<2$. 

While both \eqref{beliyau} and \eqref{basic} are sharp in the sense that they capture the main asymptotic
behaviour and therefore constants in these inequalities cannot be improved, one can expect that more subtle
bounds might invoke additional lower order correction terms. 
Indeed, we know that under certain conditions on the  geometry of $\Omega$ the asymptotics
\begin{equation*} \label{asympt2term}
R_\sigma(\La) \,= \,L^{cl}_{\sigma,d} \,\vo \,\La^{\sigma + \frac{d}{2}} - \frac14 \,L^{cl}_{\sigma,d-1} \,\left| \partial \Omega \right| \,\La^{\sigma + \frac{d-1}{2}} + o \lk \La^{\sigma+ \frac{d-1}{2}} \rk
\end{equation*}
holds true as $\La \to \infty$, see \cite{Ivr}. Recently there have been several results on semciclassical inequalities improving (\ref{beliyau}) with negative correction terms of lower order, reflecting the effect of the second term of the asymptotics, see \cite{Mel,W01,KVW}, and \cite{FLU} for discrete operators. 

Let us first point out a result of Melas. In \cite{Mel} he effectively showed that\footnote{This inequality is in fact the
Legendre transform of Melas' result.}
\begin{equation}\label{Melas}
R_1(\La) \,\leq \,L^{cl}_{1,d} \,\vo \lk \La - M_d \frac{\vo}{I(\Omega)} \rk^{1+\frac{d}{2}}_+
\end{equation}
holds for $\La > 0$. Again applying Laplace transformation Harrell and Hermi deduced
an improvement of Kac' inequality \cite{HaHe01}
\begin{equation}\label{MelKac}
Z(t) \,\leq \,\frac{\vo}{(4 \pi t)^{\frac{d}{2}}} \,\exp \lk - M_d \frac{\vo}{I(\Omega)} \,t \rk,
\end{equation}
where $I(\Omega) = \min_{a \in \R^d} \int_\Omega \left| x-a \right|^2 dx$ and $M_d$ is a constant depending only on the dimension. This improvement holds true for all $t > 0$ and any open set $\Omega$ with finite volume - without any conditions on the
boundary $\partial\Omega$. These authors conjecture also that \eqref{MelKac} can be improved to
\begin{equation} \label{HH}
Z(t) \,\leq \,\frac{\vo}{\left( 4 \pi t \right)^{\frac{d}{2}}} \,\exp \lk -\frac{t}{\vo^{\frac2d}} \rk
\end{equation}
for all $t>0$ and all open sets $\Omega$ of finite volume. Asymptotic considerations show that this conjecture is plausible for small $t$ as well as for large $t$. However, one should mention, that neither the correction term
in \eqref{Melas} is of the expected order for high energies, nor is the improvement \eqref{MelKac} or even
the conjecture \eqref{HH} of correct order for small $t>0$.

To derive universal bounds on $Z(t)$ like \eqref{HH} depending only on the volume of $\Omega$ and not including any further geometrical information one can employ an isoperimetric result due to Luttinger \cite{Lut}. He shows that Steiner-symmetrization of an open set $\Omega$ increases the trace of the heat kernel in this set. Thus for any open set $\Omega \subset \R^d$ with finite volume the inequality
\begin{equation} \label{Luttinger}
Z(t) \, \leq \, Z^*(t)
\end{equation}
holds true for all $t > 0$, where $Z^*(t)$ denotes the trace of the heat kernel in the ball $B \subset \R^d$ with the same volume as $\Omega$.

Here we prove a refined universal bound on $Z(t)$ reflecting the correct asymptotic properties. To this end we shall follow the approach in \cite{W01}. There a Berezin-Li-Yau type bound on $R_\sigma$ for $\sigma\geq 3/2$ 
with a correction term of the expected order has been found, see inequality \eqref{bly2} below. Using the same method we prove a refined Berezin-Li-Yau inequality, see Proposition \ref{blysum}, that gives rise to an improved bound on $Z(t)$ applicable to any open set $\Omega$ with finite volume. This bound decays exponentially as $t$ tends to infinity and contains a negative correction term of correct order as $t$ tends to zero.

Moreover, we can consider unbounded domains $\Omega \subset \R^d$ with infinite volume. While the results of Kac and Luttinger must fail for such domains, we show that under appropriate conditions on $\Omega$ our refined inequalities can still be applied and give order-sharp upper bounds.

This paper is structured as follows: In section \ref{mainres} we state the main results.
Then in section \ref{notation} we provide some auxiliary notation and auxiliary results including improved Berezin-Li-Yau inequalities.  
In section \ref{secthmheat} we prove Theorem \ref{thmheat} and compare this result to other bounds on $Z(t)$. In section \ref{infvol} we discuss some applications to unbounded domains and domains with infinite volume. Finally, in section \ref{secthmriesz} we apply a method by M. Aizenmann and E. H. Lieb \cite{AiLi} to the results from section \ref{notation} in order to prove refined bounds on the eigenvalue means $R_\sigma(\La)$.

We thank Rupert L. Frank for helpful discussions and in particular for indicating the result of J. M. Luttinger.

\section{Main Results}\label{mainres}

To state the main result we have to introduce some auxiliary notation. Let $\Gamma(z)$ be the usual Gamma-function
and by
\begin{equation*}
\G(z,s_1,s_2) = \int_{s_1}^{s_2} s^{z-1} e^{-s} ds / \Gamma(z) 
\end{equation*}
we denote normed incomplete Gamma-functions. If $s_1=0$ we write $\G(z,s) = \G(z,0,s)$
and $\hat\Gamma(z,s)=1-\G(z,s) = \G(z,s,+\infty)$. 
Note that for $a > 0$ we have 
\begin{equation}\label{incplgamma}
\G (a,t) = \frac{t^a}{a \ \Gamma(a)} + O \lk t^{a+1} \rk\quad\mbox{as}\quad t \to 0+ \ \mbox{and}
\end{equation}
\begin{equation}\label{infgamma}
\hat \Gamma (a,t) = \frac{t^{a-1}}{\Gamma(a)} \exp(-t) + O \lk t^{a-2} \exp(-t) \rk \quad \mbox{as} \quad t \rightarrow \infty.
\end{equation}
Furthermore, let $B(\alpha,\beta)$ be the usual Beta-function.
By 
\begin{equation*}
\B(s_1,s_2,\alpha,\beta ) =\int_{s_1}^{s_2} s^{\alpha-1} (1-s)^{\beta-1} ds / B(\alpha,\beta)
\end{equation*}
we denote normed incomplete Beta-functions and for $s_1=0$ we write in short $\B (s,\alpha,\beta) = \B(0,s,\alpha,\beta)$ and $\hat B (s,\alpha,\beta) = 1 - \B (s,\alpha,\beta) = \B (s,1,\alpha,\beta)$.
Note that for $\alpha, \beta > 0$ we have 
\begin{equation}\label{incplbeta}
B (0, t, \alpha, \beta) = \frac{1}{\alpha} t^\alpha + O \lk t^{\alpha+1} \rk\quad\mbox{as} \quad t \rightarrow 0+ \, . 
\end{equation}
Next we remark that in view of the isoperimetric inequality by Rayleigh, Faber and Krahn \cite{Fab,Kra} on the ground state $\lambda_1$
we can always choose  
\begin{equation} \label{rafakr}
\tilde \lambda=\frac{\pi}{\Gamma \lk \frac{d}{2}+1 \rk^{2/d}} \,\frac{j_{\frac{d}{2}-1,1}^2}{\vo^{2/d}} \, \leq \, \lambda_1 \, 
\end{equation}
as a lower bound on $\lambda_1$, where $j_{k,1}$ denotes the first zero of the Bessel-function $J_k$.

For $r\in\mathbb R$ put $(r)_+=\max\{r,0\}$ and for $d\in\mathbb N$ let
\begin{equation} \label{sigma}
\s_d \,= \,\left\{ \begin{array}{ccl} 5/2 & \textnormal{if} & d = 2 \\ 2 & \textnormal{if} & d = 3 \\ 3/2 & \textnormal{if} & d \geq 4 \end{array} \right. \,.
\end{equation}
Finally, let $\Omega \subset \mathbb{R}^d$ be an arbitrary open set with finite volume $\vo$.

\begin{thm} \label{thmheat}
Let $\lambda \in [ \tilde \lambda,\lambda_1]$.
For any $t > 0$ the bound
\begin{displaymath}
Z(t) \, \leq \,  \frac{\vo}{(4 \pi  t)^{\frac{d}{2}}} 
\hat\Gamma \lk \s_d + \frac{d}{2} + 1 , \lambda  t \rk 
\, - \, (R(t,\lambda))_+
\end{displaymath}
holds true with a remainder term
\begin{align*}
R(t) = c_{1,d} \, \frac{\vo^{\frac{d-1}{d}}}{(4 \pi t)^{\frac{d-1}{2}}} \, \hat \Gamma \lk \s_d + \frac{d+1}{2}, \lambda t \rk - c_{2,d} \,  \frac{\vo^{\frac{d-3}{d}}}{(4 \pi t)^{\frac{d-3}{2}}} \, \hat \Gamma \lk \s_d + \frac{d-1}{2}, \lambda t \rk \, ,
\end{align*}
where 
\begin{align*}
c_{1,d} &= \frac{B \lk \frac 12, \s_d + \frac{d+1}{2} \rk}{2}  \frac{\Gamma \lk \frac d2 +1 \rk^{\frac{d-1}{d}}}{\Gamma \lk \frac{d+1}{2} \rk} \quad \mbox{and} \\ 
c_{2,d} &=  \frac{\pi^2 (d-1)B \lk \frac 12, \s_d + \frac{d+1}{2} \rk}{96 (2\s_d+d-1)}  \frac{\Gamma \lk \frac d2 +1 \rk^{\frac{d-3}{d}}}{\Gamma \lk \frac{d+1}{2} \rk} \, .
\end{align*}
\end{thm}

\begin{rem}
Because of \eqref{incplgamma} Theorem \ref{thmheat} can then be read as
\begin{equation} \label{heatkurz}
Z(t) \, \leq \, \frac{\vo}{(4 \pi  t)^{\frac{d}{2}}} - c_{1,d} \, \frac{\vo^{\frac{d-1}{d}}}{(4 \pi t)^{\frac{d-1}{2}}}  -  r(t)
\end{equation}
with an explicit remainder term $r(t)= O ( t^{-\frac{d-3}{2}})$  as $t \to 0+$. We note that the bound captures the main asymptotic behaviour of $Z(t)$ as $t$ tends to zero: The first term equals the leading term of the short time asymptotics of $Z(t)$ and the second term shows the correct order in $t$ compared with the second term of the asymptotic expansion.

Moreover, note that in view of (\ref{infgamma}) the bound from Theorem \ref{thmheat} decays exponentially as $t$ tends to infinity. More precisely, it follows that the bound is of order $O ( t^{\s_d+1} \exp (- \tilde \lambda \, t ) )$ as $t \rightarrow \infty$. 
\end{rem}

\begin{rem}
If we choose $\lambda = \tilde \lambda$ introduced in (\ref{rafakr}) we arrive at a universal upper bound on $Z(t)$ depending only on $\vo$ and not including any explicit information on $\lambda_1$. For the explicit statement see Corollary  \ref{corheat} in section \ref{secthmheat}. This result implies the conjectured inequality \eqref{HH} for dimensions $d \leq 633$.
\end{rem}

As stated above, our proof of Theorem \ref{thmheat} relies on improved bounds for Riesz means
of eigenvalues. Let us state the corresponding result.

\begin{thm} \label{thmriesz}
Let $\lambda \in [\tilde \lambda,\lambda_1]$ and $\sigma > \s_d$ and put $\tau_\Omega = \frac{\pi^2 d^2}{\vo^{\frac 2d}}$. Then the estimate
\begin{displaymath}
R_\sigma ( \La) \,\leq \,L^{cl}_{\sigma,d} \,\vo \, \hat B \lk  \frac{\lambda}{\La}, \s_d + \frac{d}{2}+1 ,\sigma-\s_d \rk \,  \La^{\sigma+\frac{d}{2}} \,- \,
(S(\La,\lambda))_+ 
\end{displaymath}
holds true for all $\La \geq \lambda$, where
\begin{equation} \label{a1}
S(\La,\lambda) \,= \,L^{cl}_{\sigma,d-1} \,\vo^{\frac{d-1}{d}} \,\La^{\sigma+\frac{d-1}{2}} \frac{B \lk \frac{1}{2}, \s_d + \frac{d+1}{2} \rk}{2}   \hat B \lk \frac{\lambda}{\La} , \s_d + \frac{d+1}{2}, \sigma - \s_d \rk 
\end{equation}
if $\lambda \geq \tau_\Omega$ ,
\begin{equation} \label{a2}
S(\La,\lambda) \,= \,L^{cl}_{\sigma,d} \,\vo \,\La^{\sigma+\frac{d}{2}} \, \frac{1}{d} \, \hat B \lk \frac{\lambda}{\Lambda},
\s_d + \frac{d}{2}+1 , \sigma - \s_d \rk 
\end{equation}
if $\lambda < \tau_\Omega$ and $\La < \tau_\Omega$, or
\begin{align}
S(\La,\lambda) = & \, L^{cl}_{\sigma,d-1} \,\vo^{\frac{d-1}{d}}  \La^{\sigma+\frac{d-1}{2}} \frac{B \lk \frac{1}{2}, \s_d + \frac{d+1}{2} \rk }{2} \hat B \lk \frac{\tau_\Omega}{\La} , \s_d + \frac{d+1}{2}, \sigma - \s_d \rk \nonumber \\
& \ + \,L^{cl}_{\sigma,d} \, \vo \, \La^{\sigma+\frac{d}{2}} \, \frac{1}{d} \,\B \lk \frac{\lambda}{\Lambda}, \frac{\tau_\Omega}
{\La},\s_d + \frac{d}{2}+1 , \s - \s_d \rk \label{a3}\,,
\end{align}
if $\lambda < \tau_\Omega$ and $\La \geq \tau_\Omega$.
\end{thm}

\begin{rem}
Again we can choose $\lambda$ as in \eqref{rafakr} and we arrive at a universal bound depending only on $\vo$.
\end{rem}

\begin{rem}
In view of (\ref{incplbeta}) Theorem \ref{thmriesz} can be read as
\begin{equation*} 
R_\sigma (\La) \, \leq \,L^{cl}_{\sigma,d} \, \vo \, \La^{\sigma+\frac{d}{2}}  - \frac12  \, B \lk \frac12 ,\sigma_d+\frac{d+1}{2} \rk \, L^{cl}_{\sigma,d} \, \vo^{\frac{d-1}{d}} \, \La^{\sigma+\frac{d-1}{2}}  +   s(\La) 
\end{equation*}
with an explicit remainder term $s(\La) = O \lk \La^{-1} \rk$ as $\La \rightarrow \infty$.
\end{rem}

\section{Notation and auxiliary results} \label{notation}
Fix a Cartesian coordinate system in $\mathbb{R}^d$ and write 
$x = (x',x_d) \in \mathbb{R}^{d-1} \times \mathbb{R}$ for $x \in \mathbb{R}^d$. 
For a given $\Lambda > 0$ define
$$l_\Lambda \, = \, \pi\Lambda^{-\frac{1}{2}}.$$

Now consider an open set $\Omega\subset\mathbb{R}^d$.
Each section $\Omega(x') = \{x_d \in \mathbb{R} \, : \, (x',x_d) \in \Omega \}$ 
is a one-dimensional open set and
consists of at most countably many open disjoint intervals $J_k(x')$, $k =1,\dots,N(x') \leq \infty$. 
Let $\kappa(x',\Lambda) \subset \mathbb{N}$ be the subset of all those indices $k$, 
for which the corresponding interval $J_k(x')$ is strictly longer than $l_\Lambda$. 
The number these indices is denoted by $\chi(x',\Lambda)$. 
Put
\begin{equation*}
\Omega_\Lambda(x') = \bigcup_{k \in \kappa(x',\Lambda)} J_k(x') \quad\mbox{and}\quad
\Omega_\La =  \bigcup_{x' \in \mathbb{R}^{d-1}} \{ x' \} \times \Omega_\La(x')\,.
\end{equation*}
Obviously
the set $\Omega_\La$ is the subset of $\Omega$, where $\Omega$ is 
"wide enough" in $x_d$-direction. The quantity
\begin{equation*}
d_\La(\Omega) = \int_{\mathbb{R}^{d-1}} \chi\left(x',\La \right) dx'
\end{equation*}
is an effective area of the projection of $\Omega_\La$ onto the $d-1$-dimensional
hyperplane $(x',0)$ counting also the multiplicities of the sufficiently
long intervals $J_k(x')$.

Moreover, for $\mu \geq 2$ put
\begin{equation} \label{eps}
\varepsilon(\mu) = \inf_{A \geq 1} \left( \int_0^A \textstyle \left( 1- \frac{t^2}{A^2} \right)^\mu_+ \displaystyle dt - \sum_{k\geq 1} \textstyle \left(1-\frac{k^2}{A^2} \right)^\mu_+ \displaystyle \right) > 0 \, .
\end{equation}
We are now in the position to state the improved Berezin-Li-Yau bound from \cite{W01}:

\begin{pro} \label{BLYW}
For any open domain $\Omega \subset \mathbb{R}^d$, $\sigma \geq 3/2$ and all $\Lambda > 0$ the bound 
\begin{equation} \label{bly2} 
R_\sigma(\La) \,\leq  \,L^{cl}_{\sigma,d} \,\left| \Omega_\Lambda \right| \,\Lambda^{\sigma+\frac{d}{2}} \,- \,\varepsilon \textstyle \lk \sigma + \frac{d-1}{2}  \rk \displaystyle \,L^{cl}_{\sigma, d-1} d_\Lambda(\Omega) \,\Lambda^{\sigma+\frac{d-1}{2}}
\end{equation}
holds true. 
\end{pro}

Let us state also the following result on the explicit values of $\varepsilon(\mu)$.

\begin{lem} \label{lemeps}
For all $\mu \geq 3$ we have
\begin{equation*}
\varepsilon\lk \mu \rk \,= \,\frac12  \, B \lk \frac12 , \mu +1 \rk.
\end{equation*} 
\end{lem}

\begin{proof}
In view of definition (\ref{eps})  and the identity 
$$
\int_0^A \lk 1- \frac{t^2}{A^2} \rk_+^\mu \, dt \, = \, \frac A2 B \lk \frac12, \mu +1 \rk
$$
we have to show that 
$$
\sum_{k \geq 1} \lk 1- \frac{k^2}{A^2} \rk_+^\mu \, \leq \, \frac{A-1}{2} \, B \lk \frac 12, \mu +1 \rk
$$
holds true for $\mu \geq 3$ and $A \geq 1$.

For $\mu = 3$ the claim can be checked by elementary analytic methods, since there is an explicit expression for the sum in terms of $A$ and its integer part.

To deduce the estimate for $\mu > 3$, we start with the identity \cite{AiLi}
$$
\sum_{k \geq 1} \lk 1- \frac{k^2}{A^2} \rk^\mu_+ \, = \, \frac{1}{A^{2\mu}} \frac{1}{B(4,\mu-3)} \int_0^{A^2-1} \tau^{\mu-4} \lk A^2 - \tau \rk^3 \sum_{k \geq 1} \lk 1 - \frac{k^2}{A^2-\tau} \rk_+^3  d\tau 
$$
and estimate
\begin{eqnarray*}
\sum_{k \geq 1} \lk 1- \frac{k^2}{A^2} \rk^\mu_+ & \leq & \frac{1}{A^{2\mu}} \frac{ B \lk \frac 12 , 4 \rk}{B(4,\mu-3)} \int_0^{A^2-1} \tau^{\mu-4} \lk A^2 - \tau \rk^3 \frac{\lk A^2 - \tau \rk^\frac 12 -1 }{2}  \, d\tau \\
& = & \frac{1}{2 A^{2\mu}} \frac{ B \lk \frac 12 , 4 \rk}{B(4,\mu-3)} \int_0^{A^2-1} \tau^{\mu-4} \lk \lk A^2 - \tau \rk^\frac 72 - \lk A^2 - \tau \rk^3 \rk  \, d\tau \, .
\end{eqnarray*}
If we substitute $s = \frac{\tau}{A^2}$, we see that the last integral equals
$$
A^{2\mu} \int_0^{1-A^{-2}} s^{\mu-4} \lk A (1-s)^{\frac 72} - (1-s)^3 \rk ds \, = 
$$
$$
A^{2\mu} \lk A B\lk \mu-3, \frac 92 \rk - B \lk 4, \mu-3 \rk -  \int_{1-A^{-2}}^1 s^{\mu-4} \lk A (1-s)^{\frac 72} - (1-s)^3 \rk ds \rk \, .
$$
Now we can use the identity $B \lk \mu-3,\frac 92 \rk B \lk \frac 12,4 \rk / B\lk 4, \mu-3 \rk = B \lk \frac 12,\mu+1  \rk$ and substitute $t = 1-s$ to conclude
\begin{eqnarray*}
\sum_{k \geq 1} \lk 1- \frac{k^2}{A^2} \rk^\mu_+ & \leq & \frac{A}{2} B\lk \frac 12 , \mu +1 \rk -  \frac{ B \lk \frac 12 , 4 \rk}{2B(4,\mu-3)} \\
&& \times \lk B \lk 4,\mu-3\rk - \int_0^{A^{-2}} (1-t)^{\mu-4} \, t^3 \lk 1- A \sqrt t \rk \, dt \rk  \, .
\end{eqnarray*}
It remains to remark that the inequality
$$
\frac{ B \lk \frac 12 , 4 \rk}{B(4,\mu-3)} \lk B \lk 4,\mu-3\rk - \int_0^{A^{-2}} (1-t)^{\mu-4} \, t^3 \lk 1- A \sqrt t \rk \, dt \rk \, \geq \, B \lk \frac 12 , \mu +1 \rk \, ,
$$
holds true for all $A \geq 1$, since we have equality in the case $A =1$ and since the left hand side is non-decreasing in $A \geq 1$.
\end{proof}

In fact, we shall need a modified version of Proposition \ref{BLYW}.
 
Let $p_d(x';\Omega) = \left| \Omega (x') \right|_1$ be the one-dimensional Lebesgue measure of $\Omega(x')$,
that is the aggregated length of all intervals forming  $\Omega(x')$.
Since $\Omega$ is open, the function $p_d(x';\Omega)$ is Lebesgue measurable, and we can define the distribution function
\footnote{Here $|\cdot|_{d-1}$ stands for the Lebesgue measure in the dimension $d-1$.}
\begin{equation*} 
m_d(\tau;\Omega) = \left| \left\{ x' : p_d(x';\Omega) > \tau \right\} \right|_{d-1},\quad \tau>0.
\end{equation*} 
It is non-negative, non-increasing, continuous from the right and
it satisfies the identity
\begin{equation} \label{mvolume}
\int_0^\infty m_d(\tau;\Omega) \,d\tau = \vo.
\end{equation}

We interchange now the roles of $x_d$ and $x_i$ for $i = 1, \dots, d-1$
and introduce in the same way the distribution functions $m_i(\cdot;\Omega)$  for $\Omega$ measured along the $x_i$-axes. Finally, put
\begin{equation*}
M_i(y;\Omega) \, = \, \int_0^y m_i(\tau;\Omega) \, d\tau \quad\mbox{for}\quad  i = 1, \dots, d.
\end{equation*}
With this notation we can formulate a result similar to (\ref{bly2}):

\begin{pro} \label{blysum}
For any open domain $\Omega \subset \R^d$, $\sigma \geq 3/2$ and all $\Lambda > 0$
\begin{equation*} 
R_{\s} (\La) \,\leq \,L^{cl}_{\s,d} \, \int_{\frac{\pi}{\sqrt{\La}}}^\infty m_i(\tau;\Omega) \, d\tau \, \La^{\s+\frac{d}{2}}  + \delta_{\s,d} \ m_i \lk \frac{\pi}{\sqrt{\La}};\Omega \rk \,\La^{\s + \frac{d-1}{2}}
\end{equation*}
holds true for $i=1,\dots,d$ 
with $\delta_{\s,d} = \pi L^{cl}_{\s,d} - \varepsilon\lk \s+\frac{d-1}{2} \rk L^{cl}_{\s,d-1}$.
\end{pro}

\begin{rem}
Note that in the case of $\varepsilon\lk \s+\frac{d-1}{2} \rk$ = $\frac{1}{2} \,B \lk \s + \frac{d+1}{2}, \frac{1}{2} \rk$ we have $\delta_{\s,d} = 0$. In view of Lemma \ref{lemeps} this occurs 
if $\sigma + \frac{d-1}{2} \geq 3$,
in particular,
if $\sigma=\sigma_d$, with $\sigma_d$ introduced in \eqref{sigma}.
\end{rem}

\begin{rem}
For domians $\Omega$ with finite volume (\ref{mvolume}) yields 
\begin{equation*}
\int_{\frac{\pi}{\sqrt{\La}}}^\infty m_i(\tau;\Omega) d\tau = \vo - M_i \lk \frac{\pi}{\sqrt{\La}}; \Omega \rk \, .
\end{equation*}
Thus we arrive at
\begin{equation*}
R_{\s} (\La) \,\leq \,L^{cl}_{\s,d} \lk \vo - M_i \lk \frac{\pi}{\sqrt{\La}}; \Omega \rk \rk \La^{\s+\frac{d}{2}} + \delta_{\s,d} \,m_i \lk \frac{\pi}{\sqrt{\La}};\Omega \rk \,\La^{\s + \frac{d-1}{2}}
\end{equation*}
for $i = 1,\dots,d$. Averaging over all directions one claims
\begin{equation}\label{blysumeq}
R_{\s} (\La) \,\leq \,L^{cl}_{\s,d} \lk \vo - M \lk \frac{\pi}{\sqrt{\La}}; \Omega \rk \rk \La^{\s+\frac{d}{2}} + \delta_{\s,d} \,m \lk \frac{\pi}{\sqrt{\La}};\Omega \rk \,\La^{\s + \frac{d-1}{2}},
\end{equation}
where
\begin{eqnarray*}
m(t;\Omega) &=& \frac{1}{d} \lk m_1(t;\Omega) + \cdots + m_d(t;\Omega) \rk\,,\\
M(y;\Omega) &=& \frac{1}{d} \lk M_1(y;\Omega) + \cdots + M_d(y;\Omega) \rk=\int_0^y m(t;\Omega) dt\,.
\end{eqnarray*}
\end{rem}

Although Proposition \ref{blysum} is, in general, not as sharp as \eqref{bly2}, we cannot deduce
it directly quoting  Proposition \ref{BLYW},  but we have to modify the respective proof from 
\cite{W01}, which relies on operator-valued Lieb-Thirring inequalities from \cite{LaWe}.

\begin{proof}[Proof of Proposition \ref{blysum}] 
Consider the quadratic form
\begin{equation*}
\left\| \nabla u \right\|^2_{L^2(\Omega)} - \La \left\|u \right\|^2_{L^2(\Omega)} \,= \,\left\| \nabla' u \right\|^2_{L^2(\Omega)} + \int_{\R^{d-1}} dx' \int_{\Omega(x')} \lk \left| \partial_{x_d} u \right|^2 - \La |u|^2 \rk dx_d
\end{equation*}
on functions $u$ from the form core $C_0^\infty (\Omega)$. Here $\nabla'$ and $\Delta'$ denote the gradient and the Laplace operator in the first
$d-1$ directions.
The functions $u(x',\cdot )$ satisfy Dirichlet boundary conditions at the endpoints of 
each interval $J_k(x')$ forming $\Omega(x')$. 
%
%
Let the bounded, non-negative operators $W_k(x',\La)$ be the negative parts\footnote{The negative part of a real number $r$ is given
by $r_-=(|r|-r)/2\geq 0$. For operators we use the same convention in the spectral sense.}
of the Sturm-Liouville Operators $-\partial_{x_d,J_k(x')}^2 - \La$ with 
Dirichlet boundary conditions on $J_k(x')$. Then 
\begin{equation*}
W(x',\La)=\oplus_{k=1}^{N(x')} W_k(x',\La)
\end{equation*}
is the negative part of
\begin{equation*}
-\partial_{x_d,\Omega(x')}^2 - \La=\oplus_{k=1}^{N(x')} \lk -\partial_{x_d,J_k(x')}^2 - \La\rk 
\end{equation*}
subject to Dirichlet boundary conditions on the
endpoints of the intervals $J_k(x')$,
$k=1,\dots,N(x')$, that is on $\partial\Omega(x')$.
Then 
\begin{displaymath}
\int_{\Omega(x')} \lk \left| \partial_{x_d} u \right|^2 - \La |u|^2 \rk dx_d \,\geq \,- \langle W u(x',\cdot ), u(x', \cdot ) \rangle_{L^2(\Omega(x'))}.
\end{displaymath}
and consequently
\begin{equation}\label{basicproof}
\left\| \nabla u \right\|^2_{L^2(\Omega)} - \La \left\|u \right\|^2_{L^2(\Omega)} \,\geq \left\| \nabla' u \right\|^2_{L^2(\Omega)} - \int_{\R^{d-1}} dx' \langle W u(x',\cdot ), u(x', \cdot ) \rangle_{L^2(\Omega(x'))}.
\end{equation}
Now we can extend this quadratic form by zero to $C_0^\infty \lk \R^d \setminus \partial \Omega \rk$, which is a form core for $\lk - \Delta_{\R^d \setminus \Omega} \rk \oplus \lk -\Delta_\Omega - \La \rk$. This operator corresponds to the left hand side of (\ref{basicproof}), while the semi-bounded form on the right hand side is closed on the larger domain $H^1 \lk \R^{d-1},L^2(\R) \rk$, where it corresponds to the operator 
\begin{equation}\label{auxop'}
-\Delta' \otimes \mathbb{I} - W(x',\La)\quad\mbox{on}\quad L^2 \lk \R^{d-1} , L^2(\R) \rk\,. 
\end{equation}
Due to the positivity of $-\Delta_{\R^d \setminus \Omega}$ the variational principle implies that for any $\sigma\geq 0$
\begin{eqnarray*}
\textnormal{Tr} \lk -\Delta_\Omega - \La \rk^{\s}_{-} & = & \textnormal{Tr} \lk \lk -\Delta_{\R^d \setminus \Omega} \rk  \oplus \lk - \Delta_\Omega - \La \rk \rk^{\s}_{-}\\
& \leq & \textnormal{Tr} \lk - \Delta' \otimes \mathbb{I} - W(x',\La) \rk_-^{\s}.
\end{eqnarray*}
We can now apply a sharp Lieb-Thirring inequality
to the Schr\"odinger operator \eqref{auxop'}
with the operator-valued potential
$-W(x',\La)$, see \cite{LaWe}, and claim that
\begin{displaymath}
\textnormal{Tr} \lk - \Delta' \otimes \mathbb{I} - W(x',\La) \rk_-^{\s} \,\leq \,L^{cl}_{\s,d-1} \int_{\R^{d-1}} 
\textnormal{Tr}\,W^{\s+\frac{d-1}{2}}(x',\La) \,dx',\quad \s\geq \frac32\,.
\end{displaymath}

Now let $p_d(x')=\sum_k |J_k(x')|_1$ be the total length of all intervals $J_k(x')$. Then
shifting these intervals and dropping intermediate Dirichlet conditions by a variational argument we see
that the $j$-th eigenvalue of $-\partial_{x_d,\Omega(x')}^2 - \La$
is not smaller than the $j$-th eigenvalue of $-\partial_{x_d,L(x')}^2-\La$ on the interval $L(x')=[0,p_d(x')]$ 
subject to Dirichlet conditions at the endpoint of this one interval only. Thus,
\begin{equation*}
\mbox{Tr}\,W^{\s+\frac{d-1}{2}}(x',\La)\leq \mbox{Tr}\,\tilde{W}^{\s+\frac{d-1}{2}}(x',\La),
\end{equation*}
where $\tilde{W}(x',\La)$ is the negative part of $-\partial_{x_d,L(x')}^2-\La$. 
The nonzero eigenvalues of $\tilde{W}(x',\La)$ are  given explicitly by
\begin{equation*}
\mu_j = \La - \frac{\pi^2 j^2}{p^2_d(x')}=\Lambda\lk 1-\frac{l_\La^2j^2}{p^2_d(x')}\rk
\quad\mbox{for}\quad  j = 1, \dots, \left[ \frac{p_d(x')}{l_\La} \right]\,.
\end{equation*}
From this we conclude that
\begin{displaymath}
\textnormal{Tr} \lk -\Delta_\Omega - \La \rk^{\s}_{-} \,
\leq \, \La^{\s+\frac{d-1}{2}} \,L^{cl}_{\s,d-1} \int_{\R^{d-1}} 
\sum_{j\geq 1}  \lk 1- \frac{ l_\La^{2}j^2}{p_d^2(x')} \rk^{\s+\frac{d-1}{2}}_+ dx'.
\end{displaymath}
Note that the right hand side of this bound vanishes if $p_d(x') \leq l_\La$. 
For $p_d(x') > l_\La$  we have in view of \eqref{eps}
\begin{displaymath}
\sum_{j \geq 1}  \lk 1- \frac{l_\La^{2}j^2}{p_d^2(x')} \rk^{\s+\frac{d-1}{2}}_+ \,\leq \,\frac{p_d(x')}{2 \, l_\La} \,B \lk \s + \frac{d+1}{2},\frac12 \rk - \varepsilon \lk \s + \frac{d-1}{2} \rk
\end{displaymath}
and therefore
\begin{eqnarray}
\textnormal{Tr} \lk -\Delta_\Omega - \La \rk^{\s}_{-} & \leq & \frac{1}{2 \pi} \,B \lk \s + \frac{d+1}{2},\frac12 \rk
\La^{\s+\frac{d}{2}} L^{cl}_{\s,d-1} \int_{x':p_d(x') > l_\La} p_d(x') \,dx' \nonumber \\
&& - \,\varepsilon \lk \s + \frac{d-1}{2} \rk \,\La^{ \s+\frac{d-1}{2} } \,L^{cl}_{\s,d-1} \int_{x':p_d(x') > l_\La} \,dx'. \label{propproof}
\end{eqnarray}
Note that
\begin{equation*}
\int_{x':p_d(x') > l_\La} \,dx'  =  m_d \lk l_\La ; \Omega \rk
\end{equation*}
and 
\begin{equation*}
\int_{x':p_d(x') > l_\La} p_d(x') \,dx' \,
 \ = \  m_d \lk l_\La ; \Omega \rk l_\La + \int_{l_\La}^\infty m_d \lk \tau; \Omega \rk d\tau \, .
\end{equation*}
Moreover, using
\begin{equation*}
\frac{1}{2 \pi} B \lk \s + \frac{d+1}{2},\frac12 \rk L^{cl}_{\s,d-1} = L^{cl}_{\s,d},
\end{equation*}
we insert the identities above into \eqref{propproof} and arrive at
\begin{eqnarray*}
R_{\s}(\La) \ = \ \textnormal{Tr}(-\Delta-\Lambda)_{-}^{\s}  & \leq & L^{cl}_{\s,d} \, \Lambda^{\s+\frac{d}{2}}
\lk m_d \lk l_\La ; \Omega \rk l_\La + \int_{l_\La} m_d(\tau;\Omega) d\tau \rk
 \nonumber \\
&&  - \,\varepsilon \lk \s + \frac{d-1}{2} \rk L^{cl}_{\s,d-1} \,m_d \lk l_\La;\Omega\rk \Lambda^{\s+\frac{d-1}{2}}\,. 
\end{eqnarray*}
In view of $l_\La=\pi\Lambda^{-1/2}$ this yields
\begin{equation*}
R_{\s}(\La)  \leq  L^{cl}_{\s,d} \int_{\frac{\pi}{\sqrt{\La}}} m_d(\tau;\Omega) d\tau \, \La^{\s+\frac{d}{2}} + 
\delta_{\sigma,d} \, m_d \lk l_\La ; \Omega \rk \La^{\s+\frac{d-1}{2}},
\quad \sigma\geq \frac32 . 
\end{equation*}
Interchanging the roles of $x_d$ and $x_i$ we find the respective inequalities for 
any direction $i = 1,\dots, d$.
\end{proof}

In order to derive universal bounds on $R_\sigma$ independent from $M$, in particular to prove Theorem \ref{thmriesz}, one needs bounds on $M(y)$. Identity \eqref{mvolume} immediately implies 
\begin{equation}\label{eq:uppbd:momega}
M(y;\Omega)=\int_0^y m(\tau;\Omega)d\tau \leq \int_0^{\infty} m(\tau;\Omega)d\tau=
\vo\quad\mbox{for all}\quad 0 < y < \infty\,. 
\end{equation}
To prove a lower bound we first need an auxiliary result concerning rearrangements of $\Omega$. For $\Omega \subset \R^d$, $d \geq 2$, fix a Cartesian coordinate system $(x',x_d) \in \R^{d-1} \times \R$. Again put
\begin{equation*}
p_d(x';\Omega) = \left| \left\{ x_d : (x',x_d) \in \Omega \right\} \right|_1 = \left| \Omega(x') \right|_1
\end{equation*}
and for $\tau > 0$
\begin{equation*}
\Omega^*(\tau) = \left\{x' : p_d(x';\Omega) > \tau \right\} \subset \R^{d-1}.
\end{equation*}
This is a non-increasing set function, that means $\Omega^*(\tau_1) \supset \Omega^*(\tau_2)$ for $0 < \tau_1 \leq \tau_2$.
Let 
\begin{equation} \label{omstar}
\Omega^* = \cup_{\tau > 0} \, \Omega^*(\tau) \times \{\tau\} \subset \R^d
\end{equation}
be a non-increasing rearrangement of $\Omega$ in the direction
of the $x_d$-coordinate. Then we have
\begin{lem} \label{mest2}
For all $i = 1, \dots, d$ and all $y > 0$
\begin{equation*}
M_i(y;\Omega) \, \geq \, M_i(y;\Omega^*).
\end{equation*}
\end{lem}

\begin{proof}
First note that  in the case $i = d$ we have by construction
$p_d(x';\Omega)=p_d(x';\Omega^*)$ and consequently
\begin{equation*}
m_d(\tau;\Omega) = m_d(\tau;\Omega^*) = \left| \Omega^*(\tau) \right|_{d-1},
\end{equation*}
what implies $M_d(y;\Omega) = M_d(y;\Omega^*)$. 

Assume now that $j = 1, \dots, d-1$. 
Put 
$$x'' = (x_1, \dots, x_{j-1}, x_{j+1}, \dots, x_{d-1} ) \in \R^{d-2}$$
and 
$$p_j(x'',x_d; \Omega) = \left| \left\{ x_j : (x',x_d) \in \Omega \right\} \right|_1.$$ 
By definition
\begin{equation*}
m_j(s;\Omega) = \left| \left\{ (x'',x_d) : p_j(x'',x_d;\Omega) >s \right\} \right|_{d-1} 
= \int_{\R^{d-2}} \hat m_j(x'',s;\Omega) \, dx''
\end{equation*} 
where 
\begin{equation*} 
\hat m_j(x'',s;\Omega)=
\left| \left\{ x_d: p_j(x'',x_d;\Omega) > s \right\} \right|_1\,,
\quad
j = 1, \dots, d-1\,. 
\end{equation*}
Hence,
\begin{equation*}
M_j(y;\Omega) = \int_0^y m_j(s;\Omega) ds = \int_{\R^{d-2}} \int_0^y \hat m _j (x'',s;\Omega)\, ds\, dx''.
\end{equation*}
Applying the same notation to $\Omega^*$  yields
\begin{equation*}
M_j(y;\Omega^*) = \int_0^y m_j(s;\Omega^*) \, ds = \int_{\R^{d-2}} \int_0^y \hat m _j (x'',s;\Omega^*) \, ds \, dx''.
\end{equation*}
If we can show that for $x'' \in \R^{d-2}$ and all $y > 0$ the inequality
\begin{equation} \label{mint2}
\int_0^y \hat m_j(x'',s;\Omega) ds \geq \int_0^y \hat m_j(x'',s;\Omega^*) ds
\end{equation}
holds true, the assertion is proven.

To establish \eqref{mint2} we consider for fixed $x'' \in \R^{d-2}$ the two-dimensional sets
\begin{equation*}
 \hat{\Omega}  \, = \, \left\{ (x_j,x_d): (x',x_d) \in \Omega \right\} \quad \mbox{and} \quad \hat{\Omega}^*  \, = \, \left\{ (x_j,x_d): (x',x_d) \in \Omega^* \right\} \,.
\end{equation*}
Note that
\begin{equation*}
p_d(x';\Omega) = \left| \{x_d:(x',x_d)\in\Omega\} \right|_1= \left| \{x_d:(x_j,x_d)\in\hat{\Omega}\} \right|_1=:\hat p_d(x_j;\hat \Omega)\,.
\end{equation*}
As above we get
\begin{equation} \label{hatpd}
\hat p_d (x_j; \hat \Omega) = \hat p_d ( x_j; \hat \Omega^*) \, .
\end{equation}
In the $j$th direction we have
$$p_j(x'',x_d; \Omega) = \left| \left\{ x_j : (x',x_d) \in \Omega \right\} \right|_1=
\left| \left\{ x_j : (x_j,x_d) \in \hat{\Omega} \right\} \right|_1=:\hat{p}_j(x_d;\hat{\Omega})$$ 
and
$$
\hat m_j(x'',s;\Omega)=
\left| \left\{ x_d: p_j(x'',x_d;\Omega) > s \right\} \right|_1=
\left| \left\{ x_d: \hat{p}_j(x_d,\hat{\Omega}) > s \right\} \right|_1=:\hat m_j(s;\hat{\Omega})\,.
$$
The corresponding notions we use also with respect to the domains
$\Omega^*$ and $\hat\Omega^*$. In contrast to the preservation of length in the $d$th direction the values 
of $\hat{p}_j(x_d;\hat{\Omega})$ and $\hat{p}_j(x_d;\hat{\Omega}^*)$
(and thus of $\hat m_j(s;\hat{\Omega})$ and $\hat m_j(s;\hat{\Omega}^*)$) do not coincide in general.

Lets examine the functions $\hat p_j (x_d;\hat \Omega^*)$ and $\hat m_j(s;\hat \Omega^*)$ in more detail. By construction of $\hat \Omega^*$, the set function $\hat \Omega^*(x_d) = \{ x_j:(x_j,x_d) \in \hat \Omega \}$ is non-increasing in $x_d > 0$ and by definition
$$\hat p_j \lk x_d ; \hat \Omega^* \rk = \left| \hat \Omega^*(x_d) \right|_1 \, .$$
Moreover, $\hat m_j(s;\hat \Omega^*)$ is the distribution function of $\hat p_j (x_d;\hat \Omega^*)$. Hence,
\begin{equation*} 
\int_0^y \hat m_j(s;\hat \Omega^*) \, ds  =  \int_{ \left\{ x_d : \hat p_j (x_d;\hat \Omega^* ) < y \right \} } \hat p_j (x_d;\hat \Omega^* ) \, dx_d  + y \, \left| \left\{ x_d : \hat p_j (x_d ; \hat \Omega^* ) \geq y \right\} \right|_1 \, .
\end{equation*}
The monotonicity of the set function $\hat \Omega^*(x_d)$ implies, that we can choose $I_y \subset \R$ with total length $y$ satisfying $I_y \subset \hat \Omega^*(x_d)$, wherever $\hat p_j(x_d; \hat \Omega^*) \geq y$.
Again, by the monotonicity of $\hat \Omega^*(x_d)$ the reverse inclusion $\hat \Omega^*(x_d) \subset I_y$ holds for all $x_d>0$ with $\hat p_j(x_d;\hat \Omega^* ) < y$.
Put 
\begin{equation*}
\hat \Omega_y^* = \bigcup_{x_d > 0} \lk \hat \Omega^* (x_d) \cap I_y \rk \times \{x_d\}  \quad \mbox{and} \quad \hat \Omega_y = \bigcup_{x_d > 0}  \lk \hat \Omega(x_d) \cap I_y \rk \times \{x_d\}  \, . 
\end{equation*}
From the above representation for $\int_0^y \hat m_j(s;\hat \Omega^*) \, ds$  we deduce
\begin{equation} \label{yvol2}
\int_0^y \hat m_j(s;\hat \Omega^*) ds = \left| \hat \Omega_y^* \right| \, .
\end{equation}
Moreover, note that for $x_j \in I_y$
\begin{eqnarray*}
\left\{ x_d : (x_j,x_d) \in \hat \Omega_y^* \right\} & = & \left\{ x_d : (x_j, x_d ) \in \hat \Omega^* \right\} \quad \mbox{and} \\
\left\{ x_d : (x_j,x_d) \in \hat \Omega_y \right\} & = & \left\{ x_d : (x_j, x_d ) \in \hat \Omega \right\} \, .
\end{eqnarray*}
In view of \eqref{hatpd} we get
\begin{equation*}
\hat p_d(x_j;\hat \Omega_y ) = \hat p_d (x_j;\hat \Omega) = \hat p_d(x_j; \hat \Omega^*) = \hat p_d (x_j;\hat \Omega^*_y )
\end{equation*}
and we conclude that
\begin{equation} \label{yvol1}
\left| \hat \Omega_y^* \right| = \left| \hat \Omega_y \right| \, .
\end{equation}

Finally, we analyse $\hat m_j (s, \hat \Omega )$. The inclusion $\hat \Omega_y \subset \hat \Omega$ implies
\begin{equation} \label{yvol3}
\int_0^y \hat m_j (s; \hat \Omega ) \, ds \geq \int_0^y \hat m_j (s; \hat \Omega_y ) \, ds \, .
\end{equation}
Moreover, by construction of $\hat \Omega_y$ we have
\begin{equation*}
\hat p_j (x_d; \hat \Omega_y ) \leq \left| I_y \right| = y 
\end{equation*}
for all $x_d > 0$ and consequently $m_j(s;\hat \Omega_y ) = 0$ for all $s \geq y$.
Using \eqref{yvol3} we conclude
\begin{equation*}
\int_0^y \hat m_j (s; \hat \Omega ) \, ds \geq \int_0^y \hat m_j (s; \hat \Omega_y ) \, ds \, = \int_0^\infty \hat m_j (s; \hat \Omega_y ) \, ds = \left| \hat \Omega_y \right| \, .
\end{equation*}
In view \eqref{yvol1} and \eqref{yvol2} we arrive at
\begin{equation*}
\int_0^y \hat m_j (s; \hat \Omega ) \, ds \geq \left| \hat \Omega_y \right| = \left| \hat \Omega_y^* \right| = \int_0^y \hat m_j(s;\hat \Omega^*) \, ds \, .
\end{equation*}
This shows that (\ref{mint2}) holds true and the proof is complete.
\end{proof}

Now we can give a lower bound on $M(y;\Omega)$:
\begin{lem} \label{mest}
For all open sets $\Omega \subset \R^d$ and all $y > 0$
\begin{equation}\label{eq:lowbd:momega}
M(y;\Omega) \,\geq \,\min \lk \frac{\vo}{d}, \,\vo^{\frac{d-1}{d}} \,y \rk.
\end{equation}
\end{lem}

\begin{proof}
We use induction in the dimension. For $d=1$ and an interval of length $\vo$ we get
\begin{displaymath}
 m(\tau;\Omega) \,= \,\left\{ \begin{array}{lr} 1 & \vo > \tau \\ 0 & \vo \leq \tau \end{array} \right.
\end{displaymath}
and therefore $M(y;\Omega) = \int_0^y m(\tau;\Omega) d\tau = \min \lk y, \vo \rk$ for all $y>0$.

Now assume $\Omega \subset \R^d$, $d \geq 2$.  For any given $j=1,\dots,d-1$ put
$$x''=(x_1,\dots,x_{j-1},x_{j+1},\dots,x_{d-1})\in\R^{d-2}$$
and let $\tilde{m}_j(s;\tilde \Omega)=|\{x'':\tilde p_j(x''; \tilde \Omega)>s\}|_{d-2}$ be the distribution function
of a set $\tilde\Omega\subset\mathbb R^{d-1}$ with respect to the $j$-th direction, where
$\tilde p_j(x'';\tilde \Omega)=|\{x_j:x'\in\tilde \Omega\}|_1$ is the total length of the section through $\tilde \Omega$
at $x''$ in the direction of the $x_j$-coordinate. 
Applying these notions to $\Omega^*$ given in \eqref{omstar} we get
\begin{eqnarray}\notag
m_j(s;\Omega^*)&=&\left|\{(x'',\tau)\in\mathbb{R}^{d-1}: p_j \lk x'',\tau; \Omega^* \rk > s \} \right|_{d-1}  \\
\notag
						  &=&\int_0^\infty |\{x''\in\mathbb{R}^{d-2}: \tilde p_j \lk x''; \Omega^*(\tau) \rk > s\}|_{d-2} \, d\tau \\
						  &=&\int_0^\infty \tilde m_j(s;\Omega^*(\tau)) \, d\tau \, , \qquad j=1,\dots,d-1\,.\label{mgm}
\end{eqnarray}
Put $\tilde m(s;\tilde \Omega)=(d-1)^{-1}\sum_{j=1}^{d-1} \tilde m_j(s;\tilde\Omega)$.
By induction assumption we have
\begin{equation}\label{indasmp}
\tilde M(y;\tilde \Omega)=\int_0^y \tilde m(s;\tilde \Omega) \, ds\geq 
\min \lk \frac{|\tilde\Omega|_{d-1}}{d-1}, \,|\tilde\Omega|_{d-1}^{\frac{d-2}{d-1}} \,y \rk\,,
\quad y>0.
\end{equation}
Next note that in view of \eqref{mgm}
\begin{eqnarray*}
d\cdot M(y;\Omega^*)&=&M_1(y;\Omega^*)+\cdots+M_{d-1}(y;\Omega^*)+M_{d}(y;\Omega^*)\\
&=& \int_0^y (m_1(s;\Omega^*)+\cdots+m_{d-1}(s;\Omega^*))\, ds+\int_0^y m_d(s;\Omega^*) \, ds\\
&=& (d-1)\int_0^y \int_0^\infty \tilde m(s;\Omega^*(\tau)) \, d\tau \, ds +\int_0^y m_d(s;\Omega^*)\, ds\\
&=&(d-1) \int_0^\infty \tilde M(y ;\Omega^*(\tau)) \, d\tau+\int_0^y m_d(s;\Omega^*)\, ds\,.
\end{eqnarray*}
Using \eqref{indasmp} we claim 
\begin{equation*} 
M(y;\Omega^*)\geq \frac{d-1}{d} \int_0^\infty \min \lk \frac{|\Omega^*(\tau)|_{d-1}}{d-1}, \,|\Omega^*(\tau)|_{d-1}^{\frac{d-2}{d-1}} \,y \rk
d\tau + \frac1d \int_0^y m_d(s;\Omega^*)ds\,.
\end{equation*}
We point out that $\left| \Omega^*(\tau) \right|_{d-1} = m_d \lk \tau, \Omega^* \rk$ for $\tau > 0$. 
Put
\begin{equation*}
\tau^* = \inf \left\{ \tau>0 : m_d(\tau;\Omega^*) \leq (d-1)^{d-1} y^{d-1} \right\}\,.
\end{equation*}
Then
\begin{eqnarray*} 
M(y;\Omega^*)\geq \frac{1}{d} \int_{\tau^*}^\infty m_d \lk \tau, \Omega^* \rk d\tau
&+&
 \frac1d \int_0^y m_d(\tau;\Omega^*)d\tau \\
&+& \frac{d-1}{d} y\int_0^{\tau^*}\,m_d^{\frac{d-2}{d-1}}  \lk \tau; \Omega^* \rk d\tau\,.
\end{eqnarray*}
By \eqref{mvolume} we have $\int_0^{\infty} m_d \lk \tau; \Omega^* \rk d\tau=\int_0^{\infty} m_d \lk \tau; \Omega \rk d\tau=\vo$ 
and using Lemma \ref{mest2} we estimate 
\begin{eqnarray}\notag 
M(y;\Omega) \ \geq \ M(y;\Omega^*) & \geq & \frac{\vo}{d}-\frac{1}{d} \int_0^{\tau^*} m_d \lk \tau; \Omega^* \rk d\tau
+ \frac1d \int_0^y m_d(\tau;\Omega^*)d\tau \\
&& + \, \frac{d-1}{d} y\int_0^{\tau^*}\,m_d^{\frac{d-2}{d-1}}  \lk \tau; \Omega^* \rk d\tau\,.
\label{mdd-1}
\end{eqnarray}
In particular, in the case of 
$\tau^* \leq y$ we see from the previous bound that 
$$M(y;\Omega)\geq d^{-1} \, \vo$$
and the assertion is proven. 
Hence, let us consider the remaining case $\tau^*>y$ in more detail.
For $\tau^*>y$ we have
\begin{equation*} 
m_d(y;\Omega^*) \geq (d-1)^{d-1} y^{d-1}\,.
\end{equation*}
Because of the monotonicity of $m_d$ we conclude that
\begin{equation*} 
\int_0^y m_d(\tau;\Omega^*)^{\frac{d-2}{d-1}} d\tau \,\geq \,y \,m^{\frac{d-2}{d-1}}_d(y;\Omega^*)
\quad\mbox{and}\quad
\int_0^y m_d(\tau;\Omega^*) d\tau \,\geq \,y \,m_d(y;\Omega^*)\,.
\end{equation*}
Let us rewrite inequality \eqref{mdd-1} as follows
\begin{eqnarray*}
M(y;\Omega)
\geq  \frac{\vo }{d}&+& \frac{d-1}{d} \,y \int_0^{y} m_d^{\frac{d-2}{d-1}}(\tau;\Omega^*) d\tau \label{mini} \\
&+& \frac{d-1}{d} \,y \int_y^{\tau^*} m_d^{\frac{d-2}{d-1}}(\tau;\Omega^*) d\tau 
- \frac1d\int_y^{\tau^*} m_d (\tau; \Omega^* ) d\tau\,. 
\nonumber
\end{eqnarray*}
Put $A=\int_y^{\tau^*} m_d \lk \tau; \Omega^* \rk d\tau$.
Then
\begin{equation} \label{mincond4}
0<A=\int_0^{\tau^*} m_d \lk \tau; \Omega^* \rk d\tau-\int_0^{y} m_d \lk \tau; \Omega^* \rk d\tau
\leq \vo - y\,m_d(y;\Omega^*)\,.
\end{equation} 
Moreover,
\begin{equation*} 
M(y;\Omega)
\geq  \frac{\vo }{d}+ \frac{d-1}{d} \,y^2\,m^{\frac{d-2}{d-1}}_d(y;\Omega^*)
+ \frac{d-1}{d} \,y \int_y^{\tau^*} m_d^{\frac{d-2}{d-1}}(\tau;\Omega^*) d\tau 
- \frac{A}{d}\,.
\end{equation*}
Due to the monotonicity of $m_d$ we have, in particular, 
$m_d(\tau;\Omega^*)\leq m_d(y;\Omega^*)$ for $y\leq\tau$ and
\begin{eqnarray*} 
\int_y^{\tau^*} m_d^{\frac{d-2}{d-1}}(\tau;\Omega^*) d\tau
&=& m_d^{\frac{d-2}{d-1}}(y;\Omega^*) \int_y^{\tau^*} 
\lk \frac{m_d(\tau;\Omega^*)}{m_d(y;\Omega^*)} \rk^{\frac{d-2}{d-1}}d\tau\\
&\geq & 
m_d^{\frac{d-2}{d-1}}(y;\Omega^*) \int_y^{\tau^*} 
 \frac{m_d(\tau;\Omega^*)}{m_d(y;\Omega^*)} \,d\tau = m_d^{\frac{-1}{d-1}}(y;\Omega^*) \, A\,.
\end{eqnarray*}
Thus,
\begin{equation*} 
M(y;\Omega)
\geq  \frac{\vo }{d}+ \frac{d-1}{d} \,y^2\,m^{\frac{d-2}{d-1}}_d(y;\Omega^*)
- \frac1d \lk 1- (d-1) \,y \, m_d^{\frac{-1}{d-1}}(y;\Omega^*) 
 \rk A\,.
\end{equation*}
For $\tau^* > y$ we have $1- (d-1) \,y \, m_d^{\frac{-1}{d-1}}(y;\Omega^*)>0$ and we can insert  \eqref{mincond4} in this estimate and arrive at
\begin{eqnarray*}
M(y;\Omega)
&\geq&  \frac{\vo }{d}
+ \frac{d-1}{d} \,y^2\,m^{\frac{d-2}{d-1}}_d(y;\Omega^*)\\
&& - \frac1d \lk 1- (d-1) \,y\,m_d^{\frac{-1}{d-1}}(y;\Omega^*) 
 \rk (\vo-y \, m_d(y;\Omega^*))\\
 &\geq &
  \frac{y}{d} \lk (d-1) \vo\,m_d^{\frac{-1}{d-1}}(y;\Omega^*) + \,m_d(y;\Omega^*) \rk\,.
\end{eqnarray*}
Since the function $f(m)=(d-1) \vo\,m^{\frac{-1}{d-1}}+m$ takes its minimal value for positive arguments
at $m = \vo^{\frac{d-1}{d}}$, we arrive for $y<\tau^*$ at
\begin{equation*}
M(y;\Omega)\geq \frac{y}{d}f(m_d(y;\Omega^*))\geq  \frac{y}{d}f(\vo^{\frac{d-1}{d}})=y\,\vo^{\frac{d-1}{d}}\,.
\end{equation*}
This completes the proof.
\end{proof}

\section{Proof of Theorem \ref{thmheat} and remarks} \label{secthmheat}

Let
\begin{equation*}
\mathcal{L}[f(\cdot)](t)= \int_0^\infty f(\La) e^{-\La t} d\La
\end{equation*}
be the Laplace transformation of a suitable function $f:(0,+\infty)\to\mathbb R$.
For real values of $t$
it is monotone, that means a pointwise estimate $f_1(\Lambda)\leq f_2(\Lambda)$
for all $\Lambda>0$ implies $\mathcal{L}[f_1](t)\leq\mathcal{L}[f_2](t)$ for any
$t\in\mathbb R$, for which both transformations are defined. In particular,
for $\lambda\geq 0$ and $\sigma>0$ one has
\begin{equation*}
\mathcal{L}[(\Lambda-\lambda)_+^\sigma](t) = \int_\lambda^\infty (\Lambda-\lambda)^\sigma e^{-\La t} d\La = e^{-\lambda t}t^{-\sigma-1}\Gamma(\sigma+1)\,,
\quad t>0\,.
\end{equation*}
In view of the linearity of the Laplace transformation one finds for $t>0$ and $\sigma>0$
the well-known identity
\begin{equation*}
Z(t)=\mbox{Tr}\,e^{+\Delta_\Omega t} = \sum_k e^{-\lambda_k t} =\sum_k 
\frac{t^{\sigma+1}}{\Gamma(\sigma+1)}\mathcal{L}[(\Lambda-\lambda_k)_+^\sigma](t)
=\frac{t^{\sigma+1}}{\Gamma(\sigma+1)}\mathcal{L}[R_\sigma(\Lambda)]\,.
\end{equation*}
Therefore, any bound on the Riesz means of the type
\begin{equation}\label{bdriesztype}
R_\sigma(\La)\leq f(\Lambda,\Omega)\quad \mbox{for all}\quad \Lambda>0
\end{equation}
implies a bound on the heat kernel
\begin{equation}\label{bdkactype}
Z(t)\leq \frac{t^{\sigma+1}}{\Gamma(\sigma+1)} \mathcal{L}[f(\cdot,\Omega)](t)
\end{equation}
valid for all $t>0$, for which the r.h.s. is defined. For example, this 
way one can deduce \eqref{basic} from \eqref{beliyau} with any $\sigma\geq 1$\,.

Next note that in view of $R_\sigma(\Lambda)=0$ for $0<\Lambda\leq\lambda_1$ 
we have in fact 
\begin{equation*}
\Gamma(\sigma+1)t^{-\sigma-1}Z(t)=\mathcal{L}[R_\sigma](t)=\mathcal{L}[R_\sigma,\lambda](t)
\quad\mbox{for any}\quad 0\leq\lambda\leq \lambda_1\,, 
\end{equation*}
where
\begin{equation*}
\mathcal{L}[f,\lambda](t)=\int_\lambda^\infty f(\Lambda)e^{-\Lambda t} d\Lambda
=e^{-\lambda t}\mathcal{L}[f(\cdot+\lambda)](t)
\,,
\quad \lambda\geq 0\,,
\end{equation*}
is the reduced Laplace transformation of a suitable function $f$.
This transformation preserves pointwise inequalities as well and from \eqref{bdriesztype}
one can deduce an improved version of \eqref{bdkactype}
\begin{equation*}
Z(t)\leq \frac{t^{\sigma+1}}{\Gamma(\sigma+1)} \mathcal{L}[f(\cdot,\Omega),\lambda](t)
\quad\mbox{for arbitrary}\quad 0\leq\lambda\leq\lambda_1\,.
\end{equation*}
Applying this bound to \eqref{beliyau} one gets the
estimate
\begin{equation}\label{Zlargetime}
Z(t)\leq \frac{\vo}{(4\pi t)^{\frac{d}{2}}}\hat{\Gamma}\lk\sigma+\frac{d}{2}+1,\lambda t\rk\,,
\quad t>0 \, , \ \sigma\geq 1 \, , \ 0\leq\lambda\leq\lambda_1\,,
\end{equation}
which already contains an exponential decay for large $t$. 
Instead of referring to the classical Berezin-Li-Yau-bound \eqref{beliyau} 
we can apply this idea also directly to the improved bound \eqref{bly2}
and claim
\begin{eqnarray}
Z(t) & \leq & \frac{t^{\s+1}}{\Gamma(\s+1)} \,L^{cl}_{\s,d} \int_{\lambda_1}^\infty \left| \Omega_\La \right| \La^{\s+\frac{d}{2}} e^{-\La t} \,d\La \nonumber \\
&& - \, \frac{t^{\s+1}}{\Gamma(\s+1)}  \,L^{cl}_{\s,d-1} \,\textstyle \varepsilon \lk \s + \frac{d-1}{2} \rk \displaystyle \int_{\lambda_1}^\infty d_\La(\Omega) \La^{\s+\frac{d-1}{2}} e^{-\La t} \,d\La, \label{zdirect}
\end{eqnarray}
where $t>0$ and $\sigma\geq\frac32$.
This bound is even sharper than the estimates presented below. But the geometric properties of $\Omega$ 
enter in a rather tricky way and cannot be simplified in a straightforward manner. Therefore we prefer to
present also a slightly weaker, but sometimes more convenient version of this bound. For that end we choose $\s = \s_d$ given in (\ref{sigma}) and apply the reduced Laplace transformation to (\ref{blysumeq}). Thus we get the following estimate valid for $\lambda \in [\tilde \lambda, \lambda_1]$ and $t>0$:
\begin{eqnarray} 
Z(t) & \leq  & \frac{\vo}{(4\pi t)^{\frac{d}{2}}}\hat{\Gamma}\lk\s_d+\frac{d}{2}+1,\lambda t\rk \nonumber \\
&& - \frac{t^{\s_d+1} }{\Gamma(\s_d+1)}L^{cl}_{\s_d,d} \int_{\lambda}^\infty M \lk \frac{\pi}{\sqrt{\La}} \rk \La^{\s_d+\frac{d}{2}} e^{- \La t} d\La \, . \label{ZM}
\end{eqnarray}
We are now in the position to provide bounds on $Z(t)$ depending only on the volume of $\Omega$. To this end we use inequality (\ref{Luttinger}) and calculate $M \lk \frac{\pi}{\sqrt{\La}} \rk$ explicitly on the ball.

\begin{pro} \label{proheat}
Let $\lambda \in [\tilde \lambda, \lambda_1]$. For any open set $\Omega \subset \R^d$ and any $t>0$ the bound
\begin{align*}
Z(t) \ \leq & \,  \frac{\vo}{(4 \pi  t)^{\frac{d}{2}}} 
\ \hat\Gamma \lk \s_d + \frac{d}{2} + 1 , \lambda  t \rk \nonumber \\ 
&\ - \, \frac{\vo}{(4 \pi  t )^{\frac{d}{2}} \,  \Gamma \lk \s_d + \frac{d}{2}+1 \rk}  \int_{\lambda \, t }^\infty e^{-s} s^{\s_d+\frac d2} \B \lk \frac{\pi^2 t}{4 R^2 s},\frac 12,\frac{d+1}{2} \rk ds
\end{align*}
holds true, where $R = R \lk \vo \rk$ is the radius of the ball $B_R \subset \R^d$ with $|B_R|=\vo$.
\end{pro}

\begin{proof}
Lets consider the ball $B_R$ and apply (\ref{ZM}) to estimate $Z^*(t)$, i.e. $Z(t)$ on $B_R$. Note that $m_i(\tau; B_R) = m(\tau;B_R)$ for $i = 1, \dots, d$ and we can choose an arbitrary coordinate system $(x',x_d) \in \R^{d-1} \times \R$.

Again put $p_d \lk x';B_R \rk = \left| \left\{ x_d : (x', x_d ) \in B_R \right\} \right| $ and note that for $\tau < 2R$ the set $\left\{ x' \in \R^{d-1} : p_d(x',B_R) > \tau \right\}$ is itself a ball in $\R^{d-1}$ with radius $\lk R^2- \tau^2/4 \rk^{\frac 12}$. Thus we find
\begin{equation*}
m(\tau; B_R ) = \left| \left\{ x' : p_d(x',B_R) > \tau \right\} \right|_{d-1} \ = \ \frac{\pi^{\frac{d-1}{2}}}{\Gamma \lk \frac{d+1}{2} \rk} R^{d-1} \lk 1 - \frac{\tau^2}{4 R^2} \rk_+^{\frac{d-1}{2}}
\end{equation*}
and
\begin{equation*}
M(y;B_R) =  \frac{\pi^{\frac{d-1}{2}}}{\Gamma \lk \frac{d+1}{2} \rk} \, R^d \, B\lk 0, \frac{y^2}{4 R^2}, \frac{1}{2}, \frac{d+1}{2} \rk  = \left| B_R \right| \, \B \lk \frac{y^2}{4 R^2}, \frac{1}{2}, \frac{d+1}{2} \rk \, .
\end{equation*}
We insert this estimate into (\ref{ZM}) and arrive at 
\begin{align*}
Z^*(t) \, \leq &\, \frac{|B_R|}{(4 \pi  t)^{\frac{d}{2}}} \hat \Gamma \lk \s_d + \frac{d}{2} + 1 , \lambda  t \rk \\
&- \frac{t^{\s_d+1}}{\Gamma (\s_d+1)} L^{cl}_{\s_d,d} \left| B_R \right| \int_\lambda^\infty \B \lk \frac{\pi^2}{4 R^2 \La}, \frac{1}{2}, \frac{d+1}{2} \rk \La^{\s_d+\frac d2} e^{-\La t} d\La.
\end{align*}
The assumption $\lambda \geq \tilde \lambda$ implies $\frac{\pi^2}{4R^2} < \lambda$ and in view of (\ref{Luttinger}) and $|B_R| = \vo$ the claimed result follows by simplifying the right hand side.
\end{proof}

We can now derive Theorem \ref{thmheat} from Proposition \ref{proheat}:

\begin{proof}[Proof of Theorem \ref{thmheat}]
The inequality 
\begin{equation*}
(1-u)^{\frac{d-1}{2}} \geq 1 - \frac{d-1}{2} u \, , \quad 0 \leq u \leq 1 \, ,
\end{equation*}
implies the estimate
\begin{align*}
\B \lk \frac{\pi^2 t}{4 R^2 s}, \frac{1}{2}, \frac{d+1}{2} \rk &\geq \frac{1}{B \lk \frac{1}{2}, \frac{d+1}{2} \rk } \int_0^{\frac{\pi^2 t}{4  R^2 s}} u^{-\frac12} \lk 1- \frac{d-1}{2} u \rk du \\
&= \frac{\Gamma \lk \frac d2 +1\rk}{\Gamma \lk \frac{d+1}{2} \rk} \lk \frac{\sqrt {\pi t}}{R \sqrt s} - \frac{(d-1)\pi^{\frac 52} t^{\frac 32}}{24 R^3 s^{\frac 32}}  \rk \, .
\end{align*}
Therefore we claim
\begin{align*}
\lefteqn{ \int_{\lambda t}^\infty e^{-s} s^{\s_d+\frac{d}{2}} \B \lk \frac{\pi^2 \, t}{4 R^2 \, s}, \frac{1}{2}, \frac{d+1}{2} \rk ds }\\
&\geq \frac{\Gamma \lk \frac d2 +1\rk}{\Gamma \lk \frac{d+1}{2} \rk} \lk \frac{\sqrt{\pi t}}{R} \Gamma \lk \s_d + \frac{d+1}{2} , \lambda t \rk - \frac{(d-1)\pi^{\frac 52} t^{\frac 32}}{24 R^3} \Gamma \lk \s_d + \frac{d-1}{2}, \lambda t \rk \rk \, .\\
\end{align*}
Inserting the last estimate into the bound from Proposition \ref{proheat} yields
\begin{equation*}
Z(t) \leq \frac{\vo}{(4 \pi  t)^{\frac{d}{2}}} \, \hat\Gamma \lk \s_d + \frac{d}{2} + 1 , \lambda  t \rk - R(t,\lambda)
\end{equation*}
with $R(t,\lambda) = r_1(t,\lambda) - r_2(t,\lambda)$ and 
\begin{align*}
r_1(t,\lambda) &= \frac{\vo}{(4 \pi  t)^{\frac{d}{2}}}  \frac{\Gamma \lk \frac d2 +1\rk}{\Gamma \lk \frac{d+1}{2} \rk}\frac{\sqrt{\pi t}}{R} \frac{\Gamma \lk \s_d + \frac{d+1}{2} , \lambda t \rk}{\Gamma \lk \s_d+\frac d2+1 \rk} \\
r_2(t,\lambda) &= \frac{\vo}{(4 \pi  t)^{\frac{d}{2}}}  \frac{\Gamma \lk \frac d2 +1\rk}{\Gamma \lk \frac{d+1}{2} \rk} \frac{(d-1)\pi^{\frac 52} t^{\frac 32}}{24 R^3} \frac{\Gamma \lk \s_d + \frac{d-1}{2}, \lambda t \rk}{\Gamma \lk \s_d + \frac d2 +1 \rk} \, .
\end{align*}
From $|B_R| = \vo$ we deduce
\begin{equation} \label{radius}
R = \frac{\vo^{\frac 1d}}{\sqrt \pi} \Gamma \lk \frac d2 +1 \rk^{\frac 1d} \, ,
\end{equation}
and get
\begin{align*}
r_1(t,\lambda) &= \frac{\vo^{\frac{d-1}{d}}}{(4 \pi t)^{\frac{d-1}{2}}} \frac{B \lk \frac 12, \s_d + \frac{d+1}{2} \rk}{2}  \frac{\Gamma \lk \frac d2 +1 \rk^{\frac{d-1}{d}}}{\Gamma \lk \frac{d+1}{2} \rk} \hat \Gamma \lk \s_d + \frac{d+1}{2}, \lambda t \rk \\
r_2(t,\lambda) &=  \frac{\vo^{\frac{d-3}{d}}}{(4 \pi t)^{\frac{d-3}{2}}} \frac{\pi^2 (d-1)B \lk \frac 12, \s_d + \frac{d+1}{2} \rk}{96 (2\s_d+d-1)}  \frac{\Gamma \lk \frac d2 +1 \rk^{\frac{d-3}{d}}}{\Gamma \lk \frac{d+1}{2} \rk} \hat \Gamma \lk \s_d + \frac{d-1}{2}, \lambda t \rk \, .
\end{align*}
To complete the proof it remains to note that in view of (\ref{Zlargetime}) we can always estimate the remainder term $R(t,\lambda)$ from above by zero.
\end{proof}


\begin{rem}
According to \eqref{rafakr} we can choose 
\begin{equation*}
\tilde\lambda =  \frac{\pi \,  j_{\frac{d}{2}-1,1}^2}{\Gamma \lk \frac{d}{2}+1 \rk^{2/d} \vo^{2/d}}
\end{equation*}
as a suitable lower bound on $\lambda_1$. With this special choice of parameter we find
\end{rem}

\begin{cor} \label{corheat}
For any open set $\Omega \subset \mathbb{R}^d$ with finite volume and all $t > 0$
\begin{equation} \label{corheatineq}
Z(t) \, \leq \,   \frac{\vo}{(4 \pi  t)^{\frac{d}{2}}}\, \hat\Gamma \lk \sigma_d+\frac{d}{2}+1 , \tilde\lambda \,  t \rk \,
\, - \, (R(t))_+
\end{equation}
holds true with 
\begin{equation*}
R(t) = c_{1,d} \frac{\vo^{\frac{d-1}{d}}}{(4 \pi t)^{\frac{d-1}{2}}}  \hat \Gamma \lk \s_d + \frac{d+1}{2}, \tilde \lambda t \rk - c_{2,d} \frac{\vo^{\frac{d-3}{d}}}{(4 \pi t)^{\frac{d-3}{2}}} \hat \Gamma \lk \s_d + \frac{d-1}{2}, \tilde \lambda t \rk 
\end{equation*}
and constants $c_{1,d}, \, c_{2,d}$ given explicitly in Theorem \ref{thmheat}.
\end{cor}

Finally, we can apply (\ref{Luttinger}) to known estimates on $Z(t)$ and compare the resulting universal bounds with the result from Corollary \ref{corheat}.

To analyse the asymptotics of $Z(t)$ for $t \rightarrow 0+$ on convex domains van den Berg  proved \cite{Ber01} that for all convex domains $D \subset \R^d$ and all $t > 0$ 
\begin{equation*}
Z(t) \leq \frac{|D|}{(4 \pi t)^{\frac d2}} - \frac{|\partial D|}{4 (4 \pi t)^{\frac{d-1}{2}}} + \frac{(d-1) \, |\partial D| \, t}{(4 \pi t)^{\frac d2} 2R},
\end{equation*}
where $\partial D$ denotes the boundary of $D$ and at each point of $\partial D$ the curvature is bounded by $\frac 1R$. To prove bounds for general domains $\Omega$ we can apply this bound to the ball. Note that
$$| \partial B_R | \, = \, d \, \pi^{\frac d2} \frac{R^{d-1}}{\Gamma \lk \frac d2 +1 \rk} \, .$$
In view of (\ref{Luttinger}) and (\ref{radius}) we find
\begin{cor} \label{corberg}
For any open domain $\Omega \subset \R^d$ and any $t>0$
\begin{equation*}
Z(t) \leq \frac{\vo}{(4 \pi t)^{\frac d2}} - \frac{d \sqrt \pi}{\Gamma \lk \frac d2 +1 \rk^{\frac 1d}} \frac{\vo^{\frac{d-1}{d}}}{4 (4 \pi t)^{\frac{d-1}{2}}} + \frac{d(d-1)}{\Gamma \lk \frac d2 + 1 \rk^{\frac 2d}} \frac{\vo^{\frac{d-2}{d}}}{8(4 \pi t )^{\frac{d-2}{2}}} .
\end{equation*}
\end{cor}

\begin{rem}
The bounds from Corollary \ref{corheat} and Corollary \ref{corberg} both capture the main asymptotic behaviour of $Z(t)$ as $t$ tends to zero. Moreover, they contain order-sharp remainder terms. Actually, in the regime $t \rightarrow 0+$ the bound form Corollary \ref{corberg} is stronger than \eqref{corheatineq}. On the other hand the bound from Corollary \ref{corberg} does not show an exponential decay as $t$ tends to infinity.
\end{rem}

Moreover, one can use the ideas of \cite{Mel} and \cite{HaHe01} to derive unviersal bounds on $Z(t)$. We can employ inequality (\ref{MelKac}) and the result of Luttinger (\ref{Luttinger}). For the ball $B_R \subset \R^d$ with $\left| B_R \right| = \vo$ the second moment $I\lk B_R \rk $ can be calculated explicitly. If we insert the result into (\ref{MelKac}) we find
\begin{equation} \label{MelKac2}
Z(t) \,\leq \,\frac{\vo}{(4 \pi t)^{\frac{d}{2}}} \,\exp \lk - \tilde M_d \frac{t}{\vo^{\frac 2d}} \rk \, ,
\end{equation}
with a constant $\tilde M_d = \frac {d+2}{d} \pi \Gamma \lk \frac d2 +1 \rk^{-\frac 2d} M_d$. For example, in dimension $d = 2$ we have $M_2 = \frac{1}{32}$, see \cite{KVW}, and we get
\begin{equation*} 
Z(t) \,\leq \,\frac{\vo}{4 \pi t} \,\exp \lk - \frac{\pi}{16} \frac{t}{\vo} \rk \, .
\end{equation*}
In general we have $\tilde M_d < 1$ and the estimate (\ref{MelKac2}) is not strong enough to imply the conjectured inequality (\ref{HH}).

But one can employ Corollary \ref{corheat} to prove (\ref{HH}) at least in low dimensions. To analyse the asymptotic behaviour of the bound from Corollary \ref{corheat} we refer to the inequalities
\begin{equation*}
\begin{array}{llll}
j_{0,1} & > 2.4 & > \frac{1}{\sqrt{\pi}} & \mbox{if} \ d=2 \\
j_{\frac 12,1} & > 3.1 & > \displaystyle  \frac{\Gamma\lk \frac 52 \rk^{\frac13}}{\sqrt{\pi}} & \mbox{if} \ d=3 \\
j_{\frac d2 -1,1} & > \displaystyle \frac d2 -1 & >  \displaystyle \frac{\Gamma\lk \frac d2+1 \rk^{\frac1d}}{\sqrt{\pi}} & \mbox{if} \ d \geq 4 \, ,
\end{array}
\end{equation*}
see \cite{AbSt}. We find
\begin{equation*}
\tilde \lambda = \frac{\pi j_{\frac{d}{2}-1,1}^2}{\Gamma \lk \frac 2d +1 \rk^{\frac 2d} \vo^{\frac 2d}} > \frac{1}{\vo^{\frac 2d}} \, .
\end{equation*}
In view of (\ref{infgamma}) we deduce that \eqref{corheatineq} is stronger than (\ref{HH}) in the limit $t \to \infty$. Moreover one can employ (\ref{heatkurz}) to show that this relation holds true also in the limit $t \to 0+$. Finally one can compare the bounds for finite values of $t$ numerically and find that \eqref{corheatineq} is stronger than (\ref{HH}) for all $t > 0$ if $d \leq 633$ and that in these dimensions conjecture \eqref{HH} holds true.

On the other hand numerical evaluations show that for dimensions $d > 633$  there exist $t > 0$ so that the bound in \eqref{HH} is smaller than the bound in \eqref{corheatineq}. Since the conjecture \eqref{HH} does not show the expected asymptotic properties we confine ourselves to this numerical discussion.

\section{Heat kernel estimates in unbounded domains} \label{infvol}

In this section we use Proposition \ref{blysum} to prove upper bounds on $Z(t)$ in unbounded domains, in particular in domains with infinite volume. In such domains, not much is known about universal bounds on $Z(t)$, see \cite{Dav02,Dav01} for results valid in a very general setting. As an example for unbounded domains $\Omega \subset \R^2$, B. Simon and and M. van den Berg introduced ``horn-shaped'' regions \cite{Sim01,Ber03}: Assume $f : [0,\infty) \rightarrow [0,\infty)$ is a non-increasing function with $\lim_{s \rightarrow \infty} f(s) = 0$ and put
\begin{equation} \label{horn}
\Omega_f = \left\{ (x,y) \in \R^2 : x > 0 \, , \, 0 < y < f (x) \right\} \, . 
\end{equation}
Then $\Omega_f$ is ``horn-shaped''. Lets state some examples where the short time asymptotics of $Z_f(t)$ can be computed explicitly. Assume $f_\mu (s) = s^{-\frac 1\mu}$, $ \mu \geq 1$. Then for $t \to 0+$ we get 
\begin{align}
Z(t ; \Omega_{f_\mu}) &= \frac{\Gamma \lk 1 + \frac \mu2 \rk \zeta(\mu)}{2 \pi^{\mu+\frac 12}} \, t^{-\frac{\mu+1}{2}} + o \lk t^{-\frac{\mu+1}{2}} \rk &\mbox{if} \ \mu > 1 \, , \label{hornasympt1} \\
Z(t;\Omega_{f_1}) &= - \frac{\ln t}{4 \pi t} + \frac{1 + \gamma - 2 \ln (2\pi)}{4 \pi t} + O \lk t^{-\frac 12} \rk &\mbox{if} \ \mu = 1 \, , \nonumber
\end{align}
where $  \zeta(\mu)$ is the Zeta function and $\gamma$ denotes Euler's constant, see \cite{Sim01} and \cite{StTr} for refined results. Moreover one can choose $f_e(s) = \exp(-2s)$ and find 
\begin{equation} \label{hornasympt2}
Z(t;\Omega_{f_e}) = \frac{1}{4 \pi t} + \frac{\ln t}{4 \sqrt{\pi t}} + O \lk t^{-\frac 12} \rk 
\end{equation}
as $t \to 0+$, see \cite{Ber02} and \cite{StTr}.

In order to derive universal bounds on $Z(t)$ in unbounded domains, let us first note that all results mentioned in the previous sections, in particular Theorem \ref{thmheat} and Corollary \ref{corheat} remain valid for unbounded domains $\Omega \subset \R^d$ as long as $\vo$ is finite. Even if the volume of $\Omega$ is infinite the estimate (\ref{zdirect}) holds true as long as  $\Omega_\La$ is finite. Moreover, one can use Proposition \ref{blysum} to estimate $R_\s(\La)$ and $Z(t)$ as long as 
\begin{equation} \label{intcond}
\int_{\frac{\pi}{\sqrt{\La}}}^\infty m_i(\tau;\Omega) \, d\tau < \infty
\end{equation}
for all $\La > 0$. This condition is satisfied for $i=d$ and a suitable choice of coordinate system $(x',x_d) \in \R^{d-1} \times \R$ whenever 
\begin{equation*}
m_d(\tau;\Omega) = o \lk \tau^{-1} \rk \ , \quad \tau \to \infty \, .
\end{equation*}
For example we can apply Proposition \ref{blysum} to horn-shaped regions introduced in \eqref{horn} with $f_\mu (s) = s^{-\frac 1\mu}$, $\mu > 0$.

\begin{thm}
For $\mu > 0$ and all $t>0$
\begin{align*}
Z(t;\Omega_{f_\mu}) \, \leq & \, \frac{4}{105 \pi^\frac 32} \frac{1}{\mu-1}  \lk \frac{2}{\pi^2} \rk^{\frac{\mu-1}{2}} t^{- \frac{\mu+1}{2}}  \, \Gamma \lk \frac \mu2 +4, \frac{\pi^2}{2} t \rk \, + \, \frac{1}{4  \pi t} \, \hat \Gamma \lk \frac 92 , \frac{\pi^2}{2}t \rk\\
&  \, + \frac{4}{105 \pi^\frac 32} \frac{\mu}{1-\mu} \lk \frac{2}{\pi^2} \rk^{\frac{1-\mu}{2\mu}}  t^{- \frac{1+\mu}{2\mu} } \,  \Gamma \lk \frac{1}{2\mu} +4, \frac{\pi^2}{2} t \rk 
\end{align*}
if $\mu \neq 1$ and 
\begin{eqnarray*}
Z(t;\Omega_{f_1}) & \leq & - \, \frac{\ln t}{4 \pi  t} \, \hat  \Gamma \lk \frac 92 , \frac {\pi^2}{2} t \rk - \frac{1}{4 \pi  t} \, (2 \ln \pi - \ln 2) \,  \hat \Gamma \lk \frac 92, \frac {\pi^2}{2} t \rk \\
&& + \, \frac{4}{105 \pi^{\frac 32} t} \, \int_{\frac{\pi^2}{2}t}^\infty s^{\frac 72} e^{-s} \ln s \, ds \, .
\end{eqnarray*}
\end{thm}

\begin{proof}
In order to apply Proposition \ref{blysum} choose a coordinate system $(x_1,x_2) \in \R^2$ rotated by $\frac \pi4$ with respect to the coordinate system $(x,y)$ used in definition \eqref{horn}. Then for $x_1 = 0$ we have
\begin{equation*}
p_2 \lk 0;\Omega_{f_\mu} \rk = \left| \left\{ x_2 : (0,x_2) \in \Omega_{f_\mu} \right\} \right|_1 = \sqrt{2}
\end{equation*}
and we find that  $m_d \lk \tau;\Omega_{f_\mu} \rk =0$ for all $\tau \geq \sqrt{2}$.
Moreover, we can estimate
\begin{align*}
p_2 \lk x_1;\Omega_{f_\mu} \rk \ & \leq \ \sqrt{2} \ f_\mu \lk \sqrt 2 \, x_1 \rk &  \mbox{if} \  x_1 > 0  & \qquad \mbox{and}\\
p_2 \lk x_1;\Omega_{f_\mu} \rk \ & \leq \ \sqrt{2} \ f_\mu^{-1} \lk \sqrt 2 \, |x_1| \rk & \mbox{if} \ x_1 < 0 & \, ,
\end{align*}
hence 
\begin{equation*}
m_d \lk \tau;\Omega_{f_\mu} \rk \, \leq \, 2^{\frac{\mu-1}{2}} \tau^{-\mu} + 2^{\frac{1-\mu}{2\mu}} \tau^{-\frac 1\mu}
\end{equation*}
for all $0< \tau < \sqrt{2}$. Inserting these estimates into the inequality from Proposition \ref{blysum} with $\s = \s_2 = 5/2$ yields 
$$R_\frac 52(\La) \, = \, 0 \quad \textnormal{for all} \ 0 < \La \leq \frac{\pi^2}{2}$$
and for $\La > \frac{\pi^2}{2}$ we get
\begin{equation*}
R_\frac 52(\La) \, \leq \, \frac{1}{14 \pi} \lk 1- \frac{1}{1-\mu} \, \pi^{1-\mu} \, (2\La)^{\frac{\mu-1}{2}} - \frac{\mu}{\mu-1} \, \pi^{1-\frac 1\mu} \, (2\La)^{\frac{1-\mu}{2\mu}} \rk \La^\frac 72
\end{equation*}
if $\mu \neq 1$ and 
\begin{equation*}
R_\frac 52(\La) \, \leq \, \frac{1}{14 \pi} \lk  \ln \La + \ln 2 - 2 \ln \pi \rk \La^\frac 72
\end{equation*}
if $\mu = 1$. Finally by applying the Laplace transformation to these inequalities and simplifying the resulting estimates on $Z(t)$ we arrive at the claimed results. 
\end{proof}

\begin{rem}
Comparing these bounds with the asymptotic result \eqref{hornasympt1} we see that the main terms capture the correct order in $t$ as $t \to 0+$. In the case $\mu = 1$ the first term contains the sharp constant and even the second term is of correct order.
\end{rem}

To generalise these considerations to higher dimensions we use slightly different notions. Assume a non-negative function $m(\tau)$ is given for $\tau > 0$, right-continuous, non-increasing and satisfying $m(\tau) = o(\tau^{-1})$ as $t \to \infty$. Choose
\begin{equation*}
f(s) = \inf \left\{ \tau > 0 : m(\tau) \leq \omega_{d-1} s^{d-1} \right\} \, ,
\end{equation*}
where $\omega_{d-1} = \pi^{\frac{d-1}{2}} \Gamma \lk \frac{d+1}{2} \rk^{-1}$ denotes the volume of the unit ball in $\R^{d-1}$, and put
\begin{equation*}
\tilde \Omega_f = \left\{ (x',x_d) \in \R^{d-1} \times \R : \left| x_d \right| < \frac 12 f \lk \left| x' \right| \rk \right\} \, .
\end{equation*}
Then $\tilde \Omega_f$ represents an example of a domain with the distribution function 
$$m_d(\tau; \tilde \Omega_{f}) \, = \, m(\tau) \, .$$
In this case, to study explicit examples we choose $f_\mu (s) = \lk \omega_{d-1} s^{d-1} \rk^{-\frac1 \mu}$. 

\begin{thm}
\label{thmhorn2}
For any $\mu > 1$ and all $t > 0$
\begin{equation*}
Z(t;\tilde \Omega_{f_\mu}) \leq \frac{1}{(4 \pi)^{\frac d2}} \frac{\pi^{1-\mu}}{\mu-1} \frac{\Gamma \lk \s_d + \frac{d+\mu+1}{2} \rk }{\Gamma  \lk \s_d +\frac d2 +1\rk } \ t^{-\frac{d-1+\mu}{2}} \,. 
\end{equation*}
\end{thm}

\begin{proof}
The definition of $\tilde \Omega_{f_\mu}$ and the choice of $f_\mu$ implies $m_d(\tau, \Omega_{f_\mu}) = \tau^{- \mu}$.
Hence, we can employ Proposition \ref{blysum} with $\s = \s_d$ and find
\begin{equation*}
R_{\s_d}(\La) \, \leq \, L^{cl}_{\s_d,d} \int_{\frac{\pi}{\sqrt{\La}}}^\infty \tau^{-\mu} \, d\tau \, \La^{\s_d+\frac d2} \, = \, L^{cl}_{\s_d,d} \frac{\pi^{1-\mu}}{\mu-1} \La^{\s_d+\frac{d+\mu-1}{2}} \, .
\end{equation*}
To this inequality we can apply the Laplace transformation and simplify the resulting bound on $Z(t)$ to arrive at the claimed result.
\end{proof}

\begin{rem}
In dimension $d=2$, according to \cite{Sim01, Ber02, StTr}, the asymptotics \eqref{hornasympt1} and \eqref{hornasympt2} are valid for $\tilde \Omega_{f}$ as well. In this case the bound from Theorem \ref{thmhorn2} reads as
\begin{equation*}
Z(t; \tilde \Omega_{f_\mu}) \, \leq \, \frac{4}{105 \, \pi^{\mu+\frac12}} \, \frac{\Gamma \lk 4 + \frac{\mu}{2} \rk }{\mu-1} \ t^{-\frac{\mu+1}{2}} \, .
\end{equation*}
In view of (\ref{hornasympt1}) this bound shows again the correct order in $t$ as $t \to 0+$. Moreover, if we compare the constants 
\begin{equation*}
b_1(\mu) := \frac{4}{105 \, \pi^{\mu + \frac 12}} \frac{\Gamma \lk 4 + \frac{\mu}{2} \rk }{\mu-1} \quad \mbox{and} \quad b'_1(\mu) := \frac{\Gamma \lk 1 + \frac \mu2 \rk \zeta(\mu)}{2 \, \pi^{\mu+\frac 12}}
\end{equation*}
from the bound above and the asymptotics (\ref{hornasympt1}) we find 
\begin{equation*}
\lim_{\mu \rightarrow 1+} \lk \frac{b_1(\mu)}{b'_1(\mu)} \rk = \lim_{\mu \rightarrow 1+} \lk \frac{(\mu+6)(\mu+4)(\mu+2)}{(\mu-1) \zeta(\mu)} \rk = 1 \, .
\end{equation*}
\end{rem}

In order to state an example for  unbounded domains with finite volume, choose $f_e (s) = \exp \lk - \omega_{d-1} s^{d-1} \rk$. In the same way as above one can show
\begin{thm}
For all $t>0$ the estimate
\begin{align*}
Z(t,\tilde \Omega_{f_e}) &\leq  \frac{1}{(4\pi t)^{\frac d2}} \hat \Gamma \lk \s_d + \frac d2 +1 , \pi^2 t \rk + \frac{\sqrt \pi \, \ln t}{4(4 \pi t)^{\frac {d-1}{2}}} \frac{\Gamma \lk \s_d + \frac{d+1}{2}, \pi^2 t\rk }{\Gamma \lk \s_d + \frac d2 +1 \rk} \\
& \quad + \frac{\sqrt{\pi} \lk \ln \pi - 1 \rk }{2 (4 \pi t)^{\frac{d-1}{2}}} \frac{\Gamma \lk \s_d + \frac{d+1}{2}, \pi^2 t \rk}{\Gamma \lk \s_d + \frac d2 +1 \rk} \\
& \quad + \frac{\sqrt{\pi}}{4 (4 \pi t)^{\frac{d-1}{2}}} \frac{1}{\Gamma \lk \s_d + \frac d2 +1 \rk}  \int_{\pi^2 t}^\infty s^{\s_d + \frac d2} \ln s \,  e^{-s} \, ds 
\end{align*}
holds true.
\end{thm}

\begin{rem}
In view of (\ref{ztasymp}) the first term of this bound is sharp in the limit $t \rightarrow 0+$ since $\left| \Omega_{f_e} \right| = 1$. In dimension $d=2$ we can use (\ref{incplgamma}) and (\ref{hornasympt2}) to point out that even the second term of the bound captures the right order in $t$ as $t$ tends to zero.
\end{rem}

\section{Proof of Theorem \ref{thmriesz}} \label{secthmriesz}

Here we use the results from section \ref{notation} to derive universal bounds with correction terms on the Riesz means $R_\s(\La)$. First we note that Proposition \ref{blysum} and Lemma \ref{mest} immediately imply the following estimate. Recall that $\tau_\Omega = d^2 \pi^2 \vo^{-\frac 2d}$ and let $\s \geq \frac 32$ satisfy
$\sigma + \frac{d-1}{2} \geq 3$,
hence $\delta_{\s,d}=0$. Then for any open domain $\Omega \subset \R^d$ and all $\La > 0$ we find
\begin{align*}
R_\s (\La) \ &\leq \ L^{cl}_{\s,d} \, \frac{d-1}{d} \, \vo \, \La^{\s+\frac d2} &\mbox{if} \ \La < \tau_\Omega  & \quad \mbox{and} \\
R_\s (\La) \ &\leq \ L^{cl}_{\s,d} \, \vo \, \La^{\s+\frac d2} - \pi \, L^{cl}_{\s,d} \, \vo^{\frac{d-1}{d}} \La^{\s+\frac{d-1}{2}} &\mbox{if} \ \La \geq \tau_\Omega & .
\end{align*}

Next we discuss, how a trick by  Aizenmann and Lieb \cite{AiLi}
can be applied to inequalities for eigenvalue means $R_\g(\La)$ with remainder terms.

\begin{lem} \label{rieszlem}
Let $\gamma > \sigma \geq \frac32$, $\lambda_1\geq\lambda\geq 0$ and 
$\Lambda \geq \lambda$. Then 
\begin{eqnarray}
\notag
R_\g (\Lambda) & \leq & L^{cl}_{\g,d} \, \vo \, \Lambda^{\g+\frac{d}{2}} \hat B \lk  \frac{\lambda}{\La}, \s + \frac{d}{2} +1 , \g -\s \rk \\
\notag
&& - \, \frac{L^{cl}_{\s,d}}{B \lk \s+1, \g- \s  \rk} \int_0^{\La - \lambda} \tau^{\g-\s-1} M \lk \frac{\pi}{\sqrt{\La-\tau}};\Omega \rk \lk \La-\tau \rk^{\s+\frac{d}{2}} d\tau\\
&& + \, \frac{\delta_{\s,d}}{B \lk \s+1, \g- \s \rk} \int_0^{\La - \lambda} \tau^{\g-\s-1} 
m \lk \frac{\pi}{\sqrt{\La-\tau}};\Omega \rk \lk \La-\tau \rk^{\s+\frac{d-1}{2}} d\tau.
\label{eq:lm4}
\end{eqnarray}
\end{lem}

\begin{proof}
We start from the well-known identity \cite{AiLi}
\begin{eqnarray} \notag
R_\gamma (\La) &=& 
\frac{1}{B \lk  \s+1, \gamma- \s \rk} \int_0^\infty \tau^{\gamma-\s-1} \, R_\s(\La - \tau) \, d\tau
\\
\label{eq:LiAi} &=& \frac{1}{B \lk \s+1, \gamma- \s \rk} \int_0^{\La - \lambda} \tau^{\gamma-\s-1} \, R_{\s}(\La - \tau) \, d\tau.
\end{eqnarray}
Here we have taken into account that $R_\s(\tilde\Lambda)=0$ for $\tilde\Lambda\leq \lambda\leq\lambda_1$.
Now we can apply Proposition \ref{blysum} and find
\begin{eqnarray*}
R_\g (\La) & \leq & \frac{L^{cl}_{\s,d} \,\vo }{B \lk \s+1, \g- \s \rk} \,\int_0^{\La - \lambda} \tau^{\g-\s-1} \,\lk \La-\tau \rk^{\s+\frac{d}{2}} \,d\tau \\
&& - \,\frac{L^{cl}_{\s,d}}{B \lk \s+1, \g- \s  \rk} \int_0^{\La - \lambda} \tau^{\g-\s-1}  M \lk \frac{\pi}{\sqrt{\La-\tau}};\Omega \rk \lk \La-\tau \rk^{\s+\frac{d}{2}} d\tau\\
&& + \,\frac{\delta_{\s,d}}{B \lk \s+1, \g- \s \rk} \int_0^{\La - \lambda} \tau^{\g-\s-1}  m \lk \frac{\pi}{\sqrt{\La-\tau}};\Omega \rk \lk \La-\tau \rk^{\s+\frac{d-1}{2}} d\tau.
\end{eqnarray*}
Finally let us evaluate the first term on the right hand side of this expression. A substitution of the integration variable
$s = \frac{\tau}{\La}$ gives 
\begin{eqnarray*}
&&\lefteqn{ \frac{L^{cl}_{\s,d} \vo}{B\lk \s+1, \g- \s \rk} \,\La^{\g+\frac{d}{2}} \,\int_0^{1 - \frac{\lambda}{\La}} s^{\g-\s-1} \,(1-s)^{\s+\frac{d}{2}} \,ds } \\
& = & \vo \,\La^{\g+\frac{d}{2}} 
L^{cl}_{\s,d}\frac{B(\g-\s,\s+\frac{d}{2}+1)}{B(\s+1,\g-\s)} 
\lk 1 - \frac{\int_0^{\frac{\lambda}{\La}} (1-t)^{\g-\s-1} t^{\s+\frac{d}{2}} dt }{B(\g-\s,\s+\frac{d}{2}+1)} \rk \\
& = & \vo \,\La^{\g+\frac{d}{2}} L^{cl}_{\g,d} \,\lk 1 - \textstyle \B \lk \frac{\lambda}{\La},\s+\frac{d}{2}+1 ,\g-\s \rk \rk.
\end{eqnarray*}
\end{proof}

If we apply this Lemma with $\sigma=\sigma_d$ then because of $\delta_{\sigma_d,d}=0$ the
last term on the right hand side of \eqref{eq:lm4} vanishes. This enables us to finish the proof of
Theorem \ref{thmriesz}.

\begin{proof}[Proof of Theorem \ref{thmriesz}]
Inequality \eqref{eq:lm4} with $\gamma >\s=\s_d$ and with a substitution 
$y = \frac{\pi}{\sqrt{\La - \tau}}$ gives
\begin{eqnarray*}
R_\g (\Lambda) &\leq &
L^{cl}_{\g,d} \,\vo \,\Lambda^{\g+\frac{d}{2}}  \hat B \lk 
 \frac{\lambda}{\La}, \s_d+\frac{d}{2}+1, \g-\s_d \rk \nonumber \\
&& -\frac{2 \pi^{2\s_d+d+2}L^{cl}_{\s_d,d}} {B \lk \s_d +1, \g - \s_d \rk}  \, \int_{\frac{\pi}{\sqrt{\La}}}^{\frac{\pi}{\sqrt{\lambda}}} \lk \La - \frac{\pi^2}{y^2} \rk^{\gamma-\s_d-1} M(y;\Omega) \,y^{-2\s_d-d-3} \,dy 
\end{eqnarray*}
for all $\Lambda \geq \lambda$. First we assume $\lambda \geq \tau_\Omega$, i.e. $\frac{\pi}{\sqrt {\lambda}} \leq \frac 1d \vo^{1/d}$. Then we have $y \,\vo^{\frac{d-1}{d}} \leq \frac{1}{d} \vo$ for all $\frac{\pi}{\sqrt{\La}} \leq y \leq \frac{\pi}{\sqrt{\lambda}}$ and in view of Lemma \ref{mest} we get
\begin{eqnarray*}
 R_\g (\Lambda) &\leq & L^{cl}_{\g,d} \,\vo \,\Lambda^{\g+\frac{d}{2}}  \hat B \lk  \frac{\lambda}{\La}, \s_d+\frac{d}{2}+1, \g-\s_d \rk  \nonumber \\
&& - \,\frac{2 \pi^{2\s_d+d+2}L^{cl}_{\s_d,d}} {B \lk \s_d +1, \g - \s_d \rk}  \,\vo^{\frac{d-1}{d}}  \int_{\frac{\pi}{\sqrt{\La}}}^{\frac{\pi}{\sqrt{\lambda}}} \lk \La - \frac{\pi^2}{y^2} \rk^{\gamma-\s_d-1} \,y^{-2\s_d-d-2} \,dy \,.
\end{eqnarray*}
If we substitute $s = \frac{\pi^2}{y^2 \La}$ and simplify 
the expression of the remainder term we arrive at
\begin{displaymath}
R_\g ( \La) \,\leq \,L^{cl}_{\g,d} \,\vo \,\La^{\g+\frac{d}{2}}  
\hat B \lk \frac{\lambda}{\La} , \s_d+\frac{d}{2}+1 , \g - \s_d \rk  - S(\La,\lambda)
\end{displaymath}
with $S(\Lambda,\lambda)$ as stated in \eqref{a1}.

Next we assume $\lambda < \tau_\Omega$ and proceed in two steps. If 
at the same time $\La<\tau_\Omega$, that means $\frac{\pi}{\sqrt{\La}} > \frac 1d \vo^{1/d}$, we have $y \,\vo^{\frac{d-1}{d}} > \frac 1d \vo$ for all $\frac{\pi}{\sqrt{\La}} < y \frac{\pi}{\sqrt{\lambda}}$ and by similar calculations as above we arrive at 
\begin{eqnarray*}
 R_\g (\Lambda) &\leq & L^{cl}_{\g,d} \,\vo \,\Lambda^{\g+\frac{d}{2}} \hat B \lk  \frac{\lambda}{\La}, \s_d+\frac{d}{2}+1, \g-\sigma_d \rk  \\
&& - \,\frac{2 \pi^{2\s_d+d+2}L^{cl}_{\s_d,d}} {B \lk \s_d +1, \g - \s_d \rk}  \,\frac{\vo}{d} \int_{\frac{\pi}{\sqrt{\La}}}^{\frac{\pi}{\sqrt{\lambda}}} \lk \La - \frac{\pi^2}{y^2} \rk^{\gamma-\s_d-1} \,y^{-2\s_d-d-3} \,dy
\end{eqnarray*}
and 
we obtain the claimed inequality with 
$S(\La,\lambda)$ given in \eqref{a2}.
On the other hand, if $\lambda < \tau_\Omega\leq \La $  we get
\begin{eqnarray*}
 R_\g (\Lambda) &\leq & L^{cl}_{\g,d} \,\vo \,\Lambda^{\g+\frac{d}{2}} \hat B \lk  \frac{\lambda}{\La}, \s_d+\frac{d}{2}+1, \g-\sigma_d \rk  \\
&& - \,\frac{2 \pi^{2\s_d+d+2}L^{cl}_{\s_d,d}} {B \lk \s_d +1, \g - \s_d \rk}  \,\vo^{\frac{d-1}{d}} \int_{\frac{\pi}{\sqrt{\La}}}^{\frac{\vo^\frac{1}{d}}{d}} \lk \La - \frac{\pi^2}{y^2} \rk^{\gamma-\s_d-1} \,y^{-2\s_d-d-2} \,dy \\
&& - \,\frac{2 \pi^{2\s_d+d+2}L^{cl}_{\s_d,d}} {B \lk \s_d +1, \g - \s_d \rk}  \,\frac{\vo}{d} \int_{\frac{\vo^\frac{1}{d}}{d}}^{\frac{\pi}{\sqrt{\lambda}}} \lk \La - \frac{\pi^2}{y^2} \rk^{\gamma-\s_d-1} \,y^{-2\s_d-d-3} \,dy.
\end{eqnarray*}
In this case after a simplification we arrive at  
$S(\La,\lambda)$ as stated in \eqref{a3}.
Finally, if we apply \eqref{eq:LiAi} directly to \eqref{beliyau} we claim
\begin{displaymath}
R_\g ( \La) \ \leq \  L^{cl}_{\g,d} \, \vo \, \La^{\g+\frac{d}{2}} \, \hat B \lk \frac{\lambda}{\La} ,\s_d+\frac{d}{2}+1,\g - \s_d \rk \,.
\end{displaymath}
Hence, in the final bound $S(\Lambda,\lambda)$ can be replaced by its positive part
$(S(\Lambda,\lambda))_+$.
\end{proof}

Again in view of (\ref{rafakr}) we can choose
\begin{equation*}
\tilde\lambda =  \frac{\pi j_{\frac{d}{2}-1,1}^2}{\Gamma \lk \frac{d}{2}+1 \rk^{2/d} \vo^{2/d}}
\end{equation*}
as a lower bound on $\lambda_1$. Thus we find

\begin{cor} \label{corriesz}
Let $\Omega \subset \R^d$ be an open set with finite volume. Then for  $\g > \s_d$ the estimate  
\begin{displaymath}
R_\g ( \La) \,\leq \,L^{cl}_{\g,d} \, \vo \, \hat B \lk  \frac{\tilde \lambda}{\La}, \s_d + \frac d2+1,\g-\s_d \rk \La^{\g+\frac{d}{2}} \,- \, (S(\La))_+ 
\end{displaymath}
holds for all $\La \geq \tilde \lambda$, where 
\begin{equation*}
S(\La) \, =  \, L^{cl}_{\g,d} \,\vo \,\La^{\g+\frac{d}{2}} \frac{1}{d} \, \hat B \lk \frac{\tilde \lambda}{\Lambda},\s_d+\frac d2 + 1, \g - \s_d \rk
\end{equation*}
if $\La < \tau_\Omega$ and
\begin{align*}
S(\La) \, = \, & L^{cl}_{\g,d-1} \,\vo^{\frac{d-1}{d}}  \La^{\g+\frac{d-1}{2}}   \frac{B \lk \frac{1}{2}, \s_d + \frac{d+1}{2} \rk}{2} \, \hat B \lk \frac{\tau_\Omega}{\La} , \s_d + \frac{d+1}{2}, \g - \s_d \rk\\
& \,+ L^{cl}_{\g,d} \,\vo \,\La^{\g+\frac{d}{2}} \, \frac{1}{d} \, \B \lk \frac{\tilde \lambda}{\Lambda}, \frac{\tau_\Omega}{ \La},\s_d+ \frac d2 + 1, \g - \s_d \rk
\end{align*}
if $\La \geq \tau_\Omega$.
\end{cor}

\begin{rem}
We can now compare this result with estimate (\ref{bly2}) from Proposition \ref{BLYW}. In both bounds the high energy asymptotics $\La \to \infty$ is dominated by the sharp first term. In view of \eqref{incplbeta} also the remainder terms show the correct order as $\La$ tends to infinity. In this limit the bound from Corollary \ref{corriesz} is stronger than (\ref{bly2}) whenever
\begin{equation*}
\varepsilon \lk \gamma + \frac{d-1}{2} \rk d_\La(\Omega) < \frac 12 \, B\lk \frac 12 , \s_d + \frac{d+1}{2} \rk \vo^{\frac{d-1}{d}}
\end{equation*}
holds true. We remark that the right hand side is independent of $\gamma$ while $\varepsilon \lk \gamma + \frac{d-1}{2} \rk$ tends to zero as $\gamma$ tends to infinity and $d_\La(\Omega)$ is bounded from above by the diameter of $\Omega$. Hence the condition above will be satisfied for large enough $\gamma$.

Moreover, the bound from Corollary  \ref{corriesz} contains the factor 
$$\hat B \lk \frac{\tilde \lambda}{\La}, \, \s_d+\frac{d}{2} +1, \gamma - \s_d \rk$$
which decays exponentially if $\La \to \tilde \lambda+$ and which improves the bound from Theorem \ref{thmriesz} in comparison to (\ref{bly2}) for values of $\La$ close to $\tilde \lambda$.
\end{rem}

\bibliography{heat100124a}

\providecommand{\bysame}{\leavevmode\hbox to3em{\hrulefill}\thinspace}
\providecommand{\MR}{\relax\ifhmode\unskip\space\fi MR }
\providecommand{\MRhref}[2]{%
  \href{http://www.ams.org/mathscinet-getitem?mr=#1}{#2}
}
\providecommand{\href}[2]{#2}
\begin{thebibliography}{KVW08}

\bibitem[AL78]{AiLi}
M.~Aizenmann and E.H. Lieb, \emph{On semi-classical bounds for eigenvalues of
  {S}chrödinger operators}, Phys. Lett. \textbf{66} (1978), 427--429.

\bibitem[AS64]{AbSt}
M.~Abramowitz and I.A. Stegun, \emph{Handbook of mathematical functions}, 1964.

\bibitem[BC86]{BrCa}
J.~Brossard and R.~Carmona, \emph{Can one hear the dimension of a fractal?},
  Comm. Math. Phys \textbf{104} (1986), no.~1, 103--122.

\bibitem[Bro93]{Bro}
R.~M. Brown, \emph{The trace of the heat kernel in {L}ipschitz domains}, Trans.
  Amer. Math. Soc. \textbf{339} (1993), no.~2, 889--900.

\bibitem[Dav85]{Dav02}
E.~B. Davies, \emph{Trace properties of the {D}irichlet laplacian}, Math. Z.
  \textbf{188} (1985), 245--251.

\bibitem[Dav89]{Dav01}
\bysame, \emph{Heat kernels and spectral theory}, Cambridge Univ. Press, 1989.

\bibitem[Fab23]{Fab}
G.~Faber, \emph{Beweis, dass unter allen homogenen {M}embranen}, Sitzungsber.
  Bayer. Akad. Wiss. München, Math.-Phys. Kl., 1923, pp.~169--172.

\bibitem[FLU02]{FLU}
J.~K. Freericks, E.~H. Lieb, and D.~Ueltschi, \emph{Segregation in the
  {F}alicov-{K}imball model}, Comm. Math. Phys. \textbf{227} (2002), no.~2,
  243--279.

\bibitem[FLV95]{FLV}
J.~Fleckinger, M.~Levitin, and D.~Vassiliev, \emph{Heat equation on the triadic
  von {K}och snowflake: Asymptotic and aumerical analysis}, Proc. London Math.
  Soc. \textbf{71} (1995), no.~3, 372--396.

\bibitem[HH07]{HaHe01}
E.~M. Harrell and L.~Hermi, \emph{On {R}iesz means of eigenvalues},
  arXiv:0712.4088v1 (2007).

\bibitem[Ivr98]{Ivr}
V.~Ivrii, \emph{Microlocal analysis and precise spectral asymptotics}, Springer
  Monographs in Mathematics, Springer-Verlag, Berlin, 1998.

\bibitem[Kac51]{Kac02}
M.~Kac, \emph{On some connections between probability theory and differential
  and integral equations}, Proceedings of the Second Berkeley Symposium on
  Mathematical Statistics and Probability, 1950, University of California
  Press, Berkeley and Los Angeles, 1951, pp.~189--215.

\bibitem[Kac66]{Kac01}
\bysame, \emph{Can one hear the shape of a drum?}, Amer. Math. Monthly
  \textbf{73} (1966), 1--23.

\bibitem[Kra25]{Kra}
E.~Krahn, \emph{Über eine von {R}ayleigh formulierte {M}inimaleigenschaft des
  {K}reises}, Math. Ann. \textbf{94} (1925), 97--100.

\bibitem[KVW08]{KVW}
H.~Kova\v{r}\'{\i}k, S.~Vugalter, and T.~Weidl, \emph{Two dimensional
  {B}erezin-{L}i-{Y}au inequalities with a correction term}, arXiv: 0802.2792v1
  (2008).

\bibitem[Lut73]{Lut}
J.~M. Luttinger, \emph{Generalized isoperimetric inequalities}, J. Math. Phys.
  \textbf{14} (1973), no.~5, 586--593.

\bibitem[LW00]{LaWe}
A.~Laptev and T.~Weidl, \emph{Sharp {L}ieb-{T}hirring inequalities in high
  dimensions}, Acta Math. \textbf{184} (2000), no.~1, 87--111.

\bibitem[Mel03]{Mel}
A.D. Mel\'{a}s, \emph{A lower bound for sums of eigenvalues of the laplacian},
  Amer. Math. Soc \textbf{131} (2003), 631--636.

\bibitem[Min54]{Min}
S.~Minakshisundaram, \emph{Eigenfunctions on {R}iemannian manifolds}, J. Indian
  Math. Soc. (N.S.) \textbf{17} (1954), 159--165.

\bibitem[MS67]{McSi}
H.~P. McKean and I.~M. Singer, \emph{Curvature and the eigenvalues of the
  laplacian}, J. Differential Geometry \textbf{1} (1967), no.~1, 43--69.

\bibitem[Sim83]{Sim01}
B.~Simon, \emph{Non-classical eigenvalue asymptotics}, J. Functional Anal.
  \textbf{53} (1983), 84--98.

\bibitem[Smi81]{Smi}
L.~Smith, \emph{The asymptotics of the heat equation for a boundary value
  problem}, Invent. Math. \textbf{63} (1981), no.~3, 467--493.

\bibitem[ST90]{StTr}
F.~Steiner and P.~Trillenberg, \emph{Refined asymptotic expansion for the
  partition function of unbounded quantum billiards}, J. Math. Phys.
  \textbf{31} (1990), no.~7, 1670--1676.

\bibitem[vdB84a]{Ber03}
M.~van~den Berg, \emph{On the spectrum of the {D}irichlet laplacian for
  horn-shaped regions in $\mathbb{R}^n$ with infinite volume}, J. Funct. Anal.
  \textbf{58} (1984), 150--156.

\bibitem[vdB84b]{Ber01}
\bysame, \emph{A uniform bound on trace $(e^{t\Delta})$ for convex regions in
  $\mathbb{R}^n$ with smooth boundaries}, Comm. Math. Phys \textbf{92} (1984),
  no.~4, 525--530.

\bibitem[vdB87]{Ber02}
\bysame, \emph{On the asymptotics of the heat equation and bounds on traces
  associated with the {D}irichlet laplacian}, J. Funct. Anal. \textbf{71}
  (1987), no.~2, 279--293.

\bibitem[Wei08]{W01}
T.~Weidl, \emph{Improved {B}erezin-{L}i-{Y}au inequalities with a remainder
  term}, Amer. Math. Soc. Transl. \textbf{225} (2008), no.~2, 253--263.

\end{thebibliography}
\bibliographystyle{amsalpha}

\vspace{2cm}

\begin{footnotesize}
\begin{center}

\texttt{UNIVERSITÄT STUTTGART, FB MATHEMATIK, PFAFFENWALDRING 57, 70569 STUTTGART, GERMANY}

\texttt{E-MAIL: leander.geisinger@mathematik.uni-stuttgart.de}

\end{center}
\end{footnotesize}

\end{document}